\documentclass{fundam}

\publyear{25}
\papernumber{9}
\volume{195}
\issue{1-4}
\theDOI{10.46298/fi.12445}

\versionForARXIV

\pdfoutput=1

\usepackage{macros}

\begin{document}

\title{An Elementary Proof of the FMP for Kleene Algebra\texorpdfstring{\thanks{Work on this article was performed while the author was affiliated with the Open Universiteit and the Institute for Logic, Language and Computation (ILLC) at the University of Amsterdam. The author was supported by the EU’s Horizon 2020 research and innovation programme under Marie Sk\l{}odowska-Curie grant agreement No. 101027412 (VERLAN).}}{}}

\author{Tobias Kapp\'{e} \\ LIACS, Leiden University \\ t.w.j.kappe@liacs.leidenuniv.nl} 
\runninghead{T. Kapp\'{e}}{An Elementary Proof of the FMP for Kleene Algebra}

\maketitle

\begin{abstract}
Kleene Algebra (KA) is a useful tool for proving that two programs are equivalent.
Because KA's equational theory is decidable, it integrates well with interactive theorem provers.
This raises the question: which equations can we (not) prove using the laws of KA\@?
Moreover, which models of KA are \emph{complete}, in the sense that they satisfy exactly the provable equations?
Kozen (1994) answered these questions by characterizing KA in terms of its \emph{language model}.
Concretely, equivalences provable in KA are exactly those that hold for regular expressions.

Pratt (1980) observed that KA is complete w.r.t.\ \emph{relational models}, i.e., that its provable equations are those that hold for any relational interpretation.
A less known result due to Palka (2005) says that \emph{finite models} are complete for KA, i.e., that provable equivalences coincide with equations satisfied by all finite KAs.
Phrased contrapositively, the latter is a \emph{finite model property} (\emph{FMP}): any unprovable equation is falsified by a finite KA\@.
Both results can be argued using Kozen's theorem, but the implication is mutual: given that KA is complete w.r.t.\ finite (resp.\ relational) models, Palka's (resp.\ Pratt's) arguments show that it is complete w.r.t.\ the language model.

We embark on a study of the different complete models of KA, and the connections between them.
This yields a novel result subsuming those of Palka and Pratt, namely that KA is complete w.r.t.\ \emph{finite relational models}.
Next, we put an algebraic spin on Palka's techniques, which yield a new elementary proof of the finite model property, and by extension, of Kozen's and Pratt's theorems.
In contrast with earlier approaches, this proof relies not on minimality or bisimilarity of automata, but rather on representing the regular expressions involved in terms of transformation automata.

\bigskip
\end{abstract}

\section{Introduction}

Kleene Algebra (KA) provides an algebraic perspective on the equivalence of regular expressions.
Central to KA and its extensions is a theorem due to Kozen~\cite{kozen-1994}, which says that every equation valid for regular expressions can be proved using the laws of KA\@.
Salomaa showed an important precursor to this result~\cite{salomaa-1966}, and several authors~\cite{conway-1971,boffa-1990,krob-1990,foster-struth-2015,kozen-silva-2020,das-doumane-pous-2018} have studied alternative axiomatizations.

Kozen's axioms and their refinements have the advantage of giving a clear way to define what makes a ``Kleene algebra'', the same way that one might define, say, a vector space.
This has paved the way for models of KA that correspond to the semantic domains of programming languages; for instance, the relations on a set (of, e.g., machine states) have the structure of a Kleene algebra.
In turn, this motivated extensions of KA such as Kleene Algebra with Tests (KAT)~\cite{kozen-1996,kozen-smith-1996,cohen-kozen-smith-1996,kozen-1997}, which can be applied to reason about programs in general~\cite{kozen-patron-2000,kozen-tseng-2008}, as well as specific programming languages~\cite{anderson-foster-guha-etal-2014,smolka-eliopoulos-foster-guha-2015}.

As an added bonus, deciding equivalence of regular expressions, and hence validity of equations by the laws of KA(T), is practically feasible~\cite{bonchi-pous-2013,pous-2014} in spite of the problem being \textsc{pspace}-complete~\cite{stockmeyer-meyer-1973}. 
This means that KA integrates well with proof assistants, where intermediate steps in a proof can be justified by asserting that they hold for all KAs, and letting the proof assistant validate that claim~\cite{braibant-pous-2012}.

Given all of this, one might ask whether the laws of KA can be specialized to their domain of application.
Specifically, are there equations that are true for all finite or relational KAs that may not be true in general?
If so, then our toolbox can be strengthened by adding these equations when appropriate.
It turns out that KA is robust to these shifts: Palka~\cite{palka-2005} showed that the equations validated by all finite models are true for languages, whereas an observation made by Pratt~\cite{pratt-1980} implies that relational models enjoy the same property.
In conjunction with Kozen's result~\cite{kozen-2001}, this implies that every equation not provable by the laws of KA is witnessed by a finite (resp.\ relational) countermodel.
These results can therefore be regarded as the \emph{finite} (resp.\ \emph{relational}) \emph{model property} (\emph{FMP}, resp.\ \emph{RMP}) for KA\@.

The proof of Kozen's completeness theorem~\cite{kozen-1994}, however, is far from trivial.
On the other hand, if one assumes either the RMP or the FMP, a reasonably straightforward argument (due to the same authors) shows that Kozen's completeness theorem can be derived.
Pondering this situation, Palka writes that ``an independent proof of [the FMP] would provide a quite different proof of the Kozen completeness theorem, based on purely logical tools''~\cite{palka-2005}, but defers the matter to future work.

Our first contribution is a reconsideration of the arguments towards these results, painting a more complete picture of KA's model theory.
In particular, this yields a novel (if unsurprising) observation that subsumes both results: equations that hold in every finite \emph{and} relational KA also hold in general.

Our second and main contribution is to provide the independent proof of the FMP hypothesized by Palka.
In contrast with completeness arguments based on minimality~\cite{kozen-1994}, bisimilarity~\cite{kozen-2001,jacobs-2006,kozen-silva-2020} or cyclic proof systems~\cite{das-doumane-pous-2018}, our arguments rely on representing regular expressions using transformation automata~\cite{mcnaughton-papert-1968}.
A central role is played by the fact that KA allows one to solve linear systems~\cite{conway-1971,backhouse-1975,kozen-1994}.

\smallskip
The remainder of this article is organized as follows.
\Cref{section:overview} defines fundamental notions.
\Cref{section:model-theory} studies the model theory of KA to arrive at a novel \emph{finite relational model property} for KA (first contribution).
\Cref{section:solutions-automata} recalls the notion of \emph{solutions to an automaton}.
\Cref{section:transformation-automata} provides a new algebraic perspective on transformation automata~\cite{mcnaughton-papert-1968}, while \Cref{section:antimirov-construction} recalls Antimirov's construction~\cite{antimirov-1996}.
\Cref{section:prove-fmp} provides a novel proof of the FMP based on the material discussed (second contribution).
Finally, \Cref{section:discussion} concludes with some discussion and suggestions for future work.

\smallskip
This article is a revised and extended version of a paper that first appeared at RAMiCS 2023~\cite{kappe-2023}.
We answer the question raised at the end of the conference version concerning the finite relational model property, and give a new overview of the model theory of KA (\Cref{section:model-theory}) in support of that result.
Moreover, the main text now includes the proofs, which have been thoroughly revised for concision and clarity.
A substantially shorter path towards the FMP has been devised, which does not require proving that an expression is equivalent to the solution of its Antimirov automaton (which used to take up the second half of \Cref{section:antimirov-construction}).
The proof of the main claim has also been revised to make the construction less cumbersome.
As a consequence of these revisions, several formal tools have also been removed.

While the proof of the finite model property remains non-trivial, we feel comfortable claiming that its complexity is comparable to that of Kozen's completeness theorem.
Our hope is that this largely self-contained manuscript can therefore also serve as a useful reference for graduate or advanced undergraduate students interested in Kleene algebra and its model theory.

\section{Preliminaries}%
\label{section:overview}

Our primary objects of study are \emph{Kleene algebras} (\emph{KAs}).
These are interesting for computer science because in many settings, program constructs such as sequencing, branching and looping can be expressed in terms of operators that satisfy the laws of Kleene algebra.
As a consequence, properties that are true in any KA can be leveraged in proofs, for instance to show that two programs are equivalent.

\begin{definition}%
\label{definition:kleene-algebra}
A \emph{Kleene algebra (KA)} is a tuple $(K, +, \cdot, {}^*, 0, 1)$, where $(K, +, \cdot, 0, 1)$ is an idempotent semiring, and for $x, y \in K$, we have that $x^* \cdot y$ is the least fixed point (in the natural order of the semiring) of the function $z \mapsto y + x \cdot z$.
In more concrete terms, $K$ is a \emph{carrier set}, $+$ and $\cdot$ are binary operators on $K$, ${}^*$ is a unary operator on $K$, and $0, 1 \in K$, satisfying for all $x, y, z \in K$:
\begin{mathpar}
x + 0 = x
\and
x + x = x
\and
x + y = y + x
\and
x + (y + z) = (x + y) + z
\\
x \cdot (y \cdot z) = (x \cdot y) \cdot z
\and
x \cdot (y + z) = x \cdot y + x \cdot z
\and
(x + y) \cdot z = x \cdot z + y \cdot z
\\
x \cdot 1 = x = 1 \cdot x
\and
x \cdot 0 = 0 = 0 \cdot x
\and
1 + x \cdot x^* = x^*
\and
x + y \cdot z \leq z \implies y^* \cdot x \leq z
\end{mathpar}
where $\leq$ denotes the order induced by $+$, that is, $x \leq y$ if and only if $x + y = y$.
This makes $\leq$ a partial order on $K$, and ensures that all operators are monotone w.r.t.\ this order.

We often denote a generic KA $(K, +, \cdot, {}^*, 0, 1)$ by its carrier $K$, and simply write $+$, $\cdot$, etc.\ for the operators and constants when their meaning can be inferred from context.
\end{definition}

\begin{remark}%
Strictly speaking, the above defines a ``weak'' or ``left-handed'' KA, in the sense that it does not require the traditional right-unrolling and right-fixpoint laws put forward by Kozen~\cite{kozen-1996}.
Indeed, dropping these two laws does not change the equational theory of KA~\cite{krob-1990,boffa-1990,das-doumane-pous-2018,kozen-silva-2020}.
This means that if an equation holds in \emph{all} KAs $K$ that furthermore satisfy for all $x, y, z \in K$ that
\begin{mathpar}
    1 + x^* \cdot x = x
    \and
    x + z \cdot y \leq z \implies x \cdot y^* \leq z
\end{mathpar}
then this equation also holds in any KA that may not satisfy those laws.
This is somewhat surprising, but useful because there exist models of KA that do not satisfy the two right-handed axioms above~\cite{kozen-tiuryn-2003}.
\end{remark}

\begin{remark}
The last axiom of KAs listed in \Cref{definition:kleene-algebra} an implication among equalities.
This makes the class of KAs a \emph{quasivariety}~\cite{burris-sankappanavar-1981}.
The implication is necessary, in the sense that there cannot be a \emph{finite} set of equivalences without implications that defines the exact same class of algebras~\cite{conway-1971,redko-1964}.
\end{remark}

When reasoning about programs, addition is used to represent non-deterministic composition, multiplication corresponds to sequential composition, the \emph{Kleene star} operator ${}^*$ implements iteration, $0$ stands in for a program that fails immediately, and $1$ is the program that does nothing and terminates successfully.
The equations of KA correspond well to what might be expected of such operators.

\smallskip
One particularly useful (and common) class of KAs is given by the following.

\begin{definition}
Let $(K, +, \cdot, {}^*, 0, 1)$ be a KA\@.
When $x \in K$ and $n \in \mathbb{N}$, we write $x^n$ for the $n$-fold product of $x$, defined inductively by $x^0 = 1$ and $x^{n+1} = x \cdot x^n$.
We say that $K$ is \emph{star-continuous} when for all $u, v, x, y \in K$, if $u \cdot x^n \cdot v \leq y$ for all $n \in \mathbb{N}$, then $u \cdot x^* \cdot v \leq y$.
\end{definition}

Star-continuity says that the Kleene star of $x$ can be regarded as the least upper bound of all $n$-fold products of $x$.
All explicitly constructed KAs in this paper satisfy star continuity; nevertheless, not all KAs are star-continuous~\cite{kozen-1990}.
Star-continuity also introduces an infinitary implication to the axioms of KA, which makes it more cumbersome to work with from an algebraic perspective.
This is why it makes sense to generalize to the wider class of all (not necessarily star-continuous) KAs.

One very natural instance of (star-continuous) Kleene algebras is given by the \emph{relational model}.

\begin{definition}[KA of relations]%
\label{definition:ka-relations}
Let $X$ be a set.
The set of relations on $X$, denoted $\mathcal{R}(X)$, can be equipped with the structure of a KA $(\mathcal{R}(X), \cup, \circ, {}^*, \emptyset, \id_X)$, where
    $\circ$ is relational composition,
    ${}^*$ is the reflexive-transitive closure operator, and
    $\id_X$ is the identity relation on $X$.
We write $\id$ for $\id_X$ when $X$ is clear from context.
\end{definition}

If $S$ is the set of states of a machine, we can interpret a program in $\mathcal{R}(S)$, i.e., in terms of relations on $S$, by mapping it to the relation that connects initial states to final states that may be reached by running the program.
This interpretation arises naturally for programs built using the operators of KA\@.

\begin{definition}[Expressions]
We fix an \emph{alphabet} of \emph{actions} $\Sigma = \{ \ltr{a}, \ltr{b}, \ltr{c}, \dots \}$.
The set of \emph{regular expressions} $\mathbb{E}$ is given by
\[
    e, f ::= 0 \mid 1 \mid \ltr{a} \in \Sigma \mid e + f \mid e \cdot f \mid e^*
\]
Given a KA $K$ and a function $h: \Sigma \to K$, we can define $\widehat{h}: \mathbb{E} \to K$ as the map that extends $h$ and commutes with all of the operators; more explicitly, $\widehat{h}$ is specified inductively by:
\begin{align*}
    \widehat{h}(0) &= 0 &
        \widehat{h}(\ltr{a}) &= h(\ltr{a}) &
        \widehat{h}(e \cdot f) &= \widehat{h}(e) \cdot \widehat{h}(f) \\
    \widehat{h}(1) &= 1 &
        \widehat{h}(e + f) &= \widehat{h}(e) + \widehat{h}(f) &
        \widehat{h}(e^*) &= \widehat{h}(e)^*
\end{align*}
\end{definition}

The experienced reader will recognize $\mathbb{E}$ as the term algebra for the signature of KA, and $\widehat{h}$ as the unique homomorphism from $\mathbb{E}$ to $K$ induced by $h$; both are common in universal algebra~\cite{burris-sankappanavar-1981}.

\begin{example}%
\label{example:square-root-ka}
Consider a programming language with non-negative integer variables $\mathsf{Var}$, and actions $\Sigma$ comprised (for all $n \in \mathbb{N}$ and $\mathtt{x}, \mathtt{y} \in \mathsf{Var}$) of \emph{assignments} of the form $\mathtt{x} \leftarrow n$, \emph{increments} denoted $\mathtt{x} \leftarrow \mathtt{x} + \mathtt{y}$ and $\mathtt{x} \leftarrow \mathtt{x} + n$, and \emph{assertions}, written as $\mathtt{x} < \mathtt{y}$ and $\mathtt{x} \geq \mathtt{y}$.

The state of a machine executing these statements is completely determined by the value of each variable, and so we choose the set of variable assignments $S = \{ \sigma: \mathsf{Var} \to \mathbb{N} \}$ as the state space.
The semantics of the actions are relations that represent their effect, encoded by $h: \Sigma \to \mathcal{R}(S)$ as
\begin{align*}
    h(\mathtt{x} \leftarrow n) &= \{ (\sigma, \sigma[n/\mathtt{x}]) : \sigma \in S \}
        & h(\mathtt{x} < \mathtt{y}) &= \{ (\sigma, \sigma) : \sigma \in S, \sigma(\mathtt{x}) < \sigma(\mathtt{y}) \} \\
    h(\mathtt{x} \leftarrow \mathtt{x} + n) &= \{ (\sigma, \sigma[\sigma(\mathtt{x}) + n/\mathtt{x}]) : \sigma \in S \}
        & h(\mathtt{x} \geq \mathtt{y}) &= \{ (\sigma, \sigma) : \sigma \in S, \sigma(\mathtt{x}) \geq \sigma(\mathtt{y}) \} \\
    h(\mathtt{x} \leftarrow \mathtt{x} + \mathtt{y}) &= \{ (\sigma, \sigma[\sigma(\mathtt{x}) + \sigma(\mathtt{y})/\mathtt{x}]) : \sigma \in S \}
\end{align*}
Here, $\sigma[n/\mathtt{x}]$ denotes the function from $\mathsf{Var}$ to $\mathbb{N}$ that assigns $n$ to $\mathtt{x}$, and $\sigma(\mathtt{y})$ to all $\mathtt{y} \neq \mathtt{x}$.

This gives a semantics $\widehat{h}: \mathbb{E} \to \mathcal{R}(S)$ for regular expressions over $\Sigma$, and allows us to express and interpret programs like the following (where we use a semicolon instead of $\cdot$ for readability):
\[
    e = \mathtt{x} \leftarrow 1 ;\, \mathtt{y} \leftarrow 0 ;\, \mathtt{i} \leftarrow 0 ;\, (\mathtt{i} < \mathtt{n} ;\, \mathtt{y} \leftarrow \mathtt{y} + \mathtt{x}  ;\,  \mathtt{x} \leftarrow \mathtt{x} + 2 ;\, \mathtt{i} \leftarrow \mathtt{i} + 1)^* ;\, (\mathtt{i} \geq \mathtt{n})
\]
In this case one can show that the semantics of $e$ is to ``compute''  the integer square root of $\mathtt{n}$ and stores it in $\mathtt{y}$, in the sense that if $(\sigma, \sigma') \in \widehat{h}(e)$, then $\sigma'(\mathtt{y})$ is the largest integer such that ${\sigma'(\mathtt{y})}^2 \leq \sigma(\mathtt{n})$.

Of course, one can build more involved programming languages based on KA;\@ for instance, NetKAT~\cite{anderson-foster-guha-etal-2014} is a programming language for specifying and reasoning about software-defined networks.
\end{example}

Two regular expressions may yield the same semantics when interpreted in a certain KA;\@ at the same time, another KA may distinguish them.
Typically, we are interested in whether two regular expressions are equated in a given KA, in all KAs, or in a certain well-defined class of KAs.

\begin{definition}
Let $e, f \in \mathbb{E}$ and let $(K, +, \cdot, {}^*, 0, 1)$ be a KA\@.
When $\widehat{h}(e) = \widehat{h}(f)$ for all $h: \Sigma \to K$, we write $K \models e = f$.
If $K \models e = f$ for each $K$ in a class of KAs $\mathfrak{K}$, then we write $\mathfrak{K} \models e = f$.
\end{definition}

We use several classes of KAs in the sequel.
We write $\mathfrak{C}$ for the class of all star-continuous KAs, $\mathfrak{R}$ for the class of all relational KAs, i.e., KAs of the form $\mathcal{R}(X)$ for some set $X$, $\mathfrak{F}$ for all KAs with finite carrier, and $\mathfrak{FR}$ for all finite relational KAs, i.e., those in both $\mathfrak{F}$ and $\mathfrak{R}$.
We write $\equiv$ for the relation on $\mathbb{E}$ where $e \equiv f$ iff $K \models e = f$ for all KAs $K$, and we use $e \leqq f$ as a shorthand for $e + f \equiv f$.

\begin{remark}
We hinted earlier that equivalence of expressions according to the laws of KA, i.e., whether $e \equiv f$, is decidable.
This does \emph{not} mean that equivalence of expressions in a \emph{particular} KA is decidable; for instance, programs interpreted in the KA $\mathcal{R}(S)$ from \Cref{example:square-root-ka} can easily simulate a two-counter machine, which means that equivalence can be used to encode the halting problem.
Rather, $e \equiv f$ is sufficient to conclude that $\widehat{h}(e) = \widehat{h}(f)$, but not vice versa.
For instance, $\widehat{h}(\mathtt{x} \leftarrow 1; \mathtt{x} \leftarrow 2) = \widehat{h}(\mathtt{x} \leftarrow 2)$, but other models, or even different interpretations inside $\mathcal{R}(S)$, do not admit this equivalence.
In truth, $\equiv$ is unaware of any interpretation given to $\mathtt{x} \leftarrow n \in \Sigma$ --- this is what we mean when we say that KA \emph{abstracts from the meaning of primitive programs}.
\end{remark}

There are a number of natural facts about $\equiv$ and $\leqq$ that we use so often that it pays to record them now.
We apply them without appealing explicitly to this lemma, and leave the proofs to the reader.
\begin{lemma}
The following hold:
\begin{enumerate}
    \setlength{\itemsep}{0em}
    \item
    The relation $\equiv$ is a congruence.
    Concretely, this means that $\equiv$ is an equivalence relation such that if $e_1 \equiv f_1$ and $e_2 \equiv f_2$, then $e_1 + e_2 \equiv f_1 + f_2$, and similarly for the other operators.
    \item
    The relation $\leqq$ is a partial order up to $\equiv$, i.e., it is a preorder and if $e \leqq f$ and $f \leqq e$, then $e \equiv f$.
    \item
    All operators are monotone w.r.t. $\leqq$.
    Specifically, if $e_1 \leqq f_1$ and $e_2 \leqq f_2$, then $e_1 + e_2 \leqq f_1 + f_2$, and analogous properties hold for the other operators.
\end{enumerate}
\end{lemma}

\noindent
Some (classes of) KAs are \emph{complete} in that their equations are exactly the ones satisfied by \emph{all} KAs; alternatively, we can think of these KAs as having enough flexibility to differentiate regular expressions that cannot be proven equal from the laws of KA alone.
As it turns out, both $\mathfrak{R}$ and $\mathfrak{F}$ have this property.

\begin{restatable}[finite model property; Palka~\cite{palka-2005}]{theorem}{restatefmp}%
\label{theorem:finite-model-property}
Let $e, f \in \mathbb{E}$.
If $\mathfrak{F} \models e = f$, then $e \equiv f$.
\end{restatable}

\begin{theorem}[relational model property; Pratt~\cite{pratt-1980,kozen-smith-1996}]%
\label{theorem:relational-model-property}
Let $e, f \in \mathbb{E}$.
If $\mathfrak{R} \models e = f$, then $e \equiv f$.
\end{theorem}

That the equations of KA can be characterized by these subclasses of KAs is already interesting, but in fact a characterization in terms of \emph{single} KA is also possible.
To elaborate, let us write $\Sigma^*$ for the set of \emph{words} (i.e., finite sequences of actions) over $\Sigma$, and $\epsilon$ for the empty word.
We \emph{concatenate} words by juxtaposition: if $w = \ltr{a}\ltr{b}$ and $x = \ltr{c}\ltr{a}$, then $wx = \ltr{a}\ltr{b}\ltr{c}\ltr{a}$.
A set of words is called a \emph{language}.
Languages $L$ and $L'$ can be concatenated pointwise, writing $L \cdot L'$ for $\{ ww' : w \in L, w' \in L' \}$.
The \emph{Kleene star} of a language $L$, denoted $L^*$, is given by $\{ w_1\cdots{}w_n : w_1, \dots, w_n \in L \}$.
Note that $\epsilon \in L^*$.

\begin{definition}
The \emph{KA of languages} $\mathcal{L}$ is $(\mathcal{P}(\Sigma^*), \cup, \cdot, {}^*, \emptyset, \{ \epsilon \})$, where $\cdot$ is language concatenation, and ${}^*$ is the Kleene star of a language as above.
We define $\ell: \Sigma \to \mathcal{P}(\Sigma^*)$ by $\ell(\mathtt{a}) = \{ \mathtt{a} \}$.
\end{definition}

Readers familiar with regular languages will recognize $\widehat{\ell}: \mathbb{E} \to \mathcal{P}(\Sigma^*)$ as the standard language semantics of regular expressions.
It is not hard to show that $\mathcal{L} \models e = f$ precisely when $e$ and $f$ denote the same regular language.
As indicated, the equations of the language model generalize to all KAs.

\begin{theorem}[language completeness; Kozen~\cite{kozen-1996}]%
\label{theorem:language-completeness}
Let $e, f \in \mathbb{E}$.
If $\mathcal{L} \models e = f$, then $e \equiv f$.
\end{theorem}

\begin{remark}
Because the theorem above characterizes equations in KA in terms of a single KA, standard arguments from universal algebra imply that the \emph{KA of regular languages} over $\Sigma$, denoted $\mathsf{Reg}(\Sigma)$ --- i.e., the image of $\widehat{\ell}$ --- is the \emph{free Kleene algebra} on $\Sigma$~\cite{burris-sankappanavar-1981}.
Specifically, this means that for any KA $K$ and $f: \Sigma \to K$, there exists a unique $\widetilde{f}: \mathsf{Reg}(\Sigma^*) \to K$ such that $\widetilde{f}(\ltr{a}) = f(\ltr{a})$ for all $\ltr{a} \in \Sigma$ that commutes with the operators, e.g., with $\widetilde{f}(L_1 \cup L_2) = \widetilde{f}(L_1) + \widetilde{f}(K_2)$.
An equivalent and even more abstract way of saying this is that the map $\mathsf{Reg}$, which sends an alphabet $\Sigma$ to the KA of regular languages over $\Sigma$, is left adjoint to the map (functor) that sends a KA to its carrier set.
\end{remark}

Equivalence of regular languages is decidable (in fact, \textsc{pspace}-complete~\cite{stockmeyer-meyer-1973} but practically feasible~\cite{bonchi-pous-2013}), and so equivalence in the language model (and, by extension, $\equiv$) is also decidable.

\section{Model theory}%
\label{section:model-theory}

The theorems above are closely connected, in that they imply one another~\cite{palka-2005,pratt-1980,kozen-smith-1996}.
In this section, we provide alternative proofs in support of this, which also yields the (seemingly new) observation that $\mathfrak{FR}$, the class of \emph{finite relational models}, is complete for KA\@.
Besides this, none of the material in this section can truly be considered novel; we include our own proofs for the sake of self-containment.

We start with the fact that interpretations of expressions inside a star-continuous KA can be connected to the interpretations of words in their language semantics.
Here, and for the remainder of this article, we abuse notation by regarding words $w \in \Sigma^*$ as regular expressions, identifying $\epsilon$ with $1$.

\begin{lemma}[{\cite[Lemma~7.1]{kozen-1992}}]%
\label{lemma:cont}
Let $K$ be a star-continuous KA, and let $h: \Sigma \to K$.
For all $e \in \mathbb{E}$ and $x \in K$, if for all $w \in \widehat{\ell}(e)$ it holds that $\widehat{h}(w) \leq x$, then we also have that $\widehat{h}(e) \leq x$.
\end{lemma}
\begin{proof}
We prove the first claim by induction on $e$, showing more generally that for all $u, v, x \in K$, if for all $w \in \widehat{\ell}(e)$ we have that $u \cdot \widehat{h}(w) \cdot v \leq x$, then $u \cdot \widehat{h}(e) \cdot v \leq x$.
First, if $e = 0$, then $\widehat{\ell}(e) = \emptyset$, and so the claim holds vacuously.
If $e = 1$ or $e = \mathtt{a}$ for some $\mathtt{a} \in \Sigma$, then $\widehat{h}(e) = \widehat{h}(w)$ for the unique $w \in \widehat{\ell}(e)$, and so $u \cdot \widehat{h}(e) \cdot v = u \cdot \widehat{h}(w) \cdot v \leq x$.
There are three inductive cases.
\begin{itemize}
    \item
    If $e = e_1 + e_2$, then for all $w \in \widehat{\ell}(e_i) \subseteq \widehat{\ell}(e)$ we have that $u \cdot \widehat{h}(w) \cdot v \leq x$.
    Hence, by induction, $u \cdot \widehat{h}(e_i) \cdot v \leq x$, and thus $u \cdot \widehat{h}(e) \cdot v = u \cdot \widehat{h}(e_1) \cdot v + u \cdot \widehat{h}(e_2) \cdot v \leq x + x = x$.

    \item
    If $e = e_1 \cdot e_2$, then for all $w_1 \in \widehat{\ell}(e_1)$ and $w_2 \in \widehat{\ell}(e_2)$ we have that $u \cdot \widehat{h}(w_1) \cdot \widehat{h}(w_2) \cdot v \leq x$.
    Thus, by induction, it follows that for all $w_1 \in \widehat{\ell}(e_1)$ we have that $u \cdot \widehat{h}(w_1) \cdot \widehat{h}(e_2) \cdot v \leq x$.
    Applying the induction hypothesis once more, we find that $u \cdot \widehat{h}(e) \cdot v = u \cdot \widehat{h}(e_1) \cdot \widehat{h}(e_2) \cdot v \leq x$.

    \item
    If $e = e_1^*$, then an argument similar to the above shows that, for all $n \in \mathbb{N}$, if $u \cdot \widehat{h}(w) \cdot v \leq x$ for all $w \in \widehat{\ell}(e_1)^n$, then $u \cdot \widehat{h}(e_1)^n \cdot v \leq x$.
    Because $w \in \widehat{\ell}(e_1)^n \subseteq \widehat{\ell}(e_1^*)$, the premise holds for each $n$, hence $u \cdot \widehat{h}(e)^n \cdot v \leq x$ for all $n \in \mathbb{N}$.
    By star-continuity, $u \cdot \widehat{h}(e^*) \cdot v = u \cdot \widehat{h}(e)^* \cdot v \leq x$.
\end{itemize}

\vspace{-7.5mm}
\end{proof}

\noindent
The following property, which connects words in the language semantics to provable equivalence, is also useful.
It almost directly implies the converse of the previous lemma.

\begin{lemma}[{\cite[Lemma~7.1]{kozen-1992}}]%
\label{lemma:embed-sem}
For all $e \in \mathbb{E}$ and $w \in \widehat{\ell}(e)$, it holds that $w \leqq e$.
\end{lemma}
\begin{proof}
We proceed by induction on $e$.
In the base, there are three cases.
If $e = 0$, then $\widehat{\ell}(e) = \emptyset$, so the claim holds vacuously; if $e = 1$ or $e = \ltr{a}$ for some $\ltr{a} \in \Sigma$, then $w = e$, so so the claim is trivial.

In the inductive step, there are three cases to consider.
First, if $e = e_1 + e_2$, then without loss of generality $w \in \widehat{\ell}(e_1)$; by induction $w \leqq e_1 \leqq e$.
Second, if $e = e_1 \cdot e_2$, then $w = w_1 w_2$ with $w_i \in \widehat{\ell}(e_i)$; by induction and monotonicity, $w = w_1 w_2 \leqq e_1 e_2 = e$.
Last, if $e = e_1^*$, then $w = w_1 \cdots w_n$ for $w_1, \dots, w_n \in \widehat{\ell}(e_1)$; by induction, $w_i \leqq e_1$ for all $1 \leq i \leq n$, meaning $w \leqq e_1^n$.
A straightforward induction on $n$ using the rule $f \equiv 1 + f \cdot f^*$ shows that $e_1^n \leqq e_1^*$ for all $n$, and so we are done.
\end{proof}

\begin{remark}
Taken together, \Cref{lemma:cont,lemma:embed-sem} can be interpreted as saying that if $K$ is star-continuous and $h: \Sigma \to K$ is a function, then $\widehat{h}(e)$ is the least upper bound (inside $K$) of all $\widehat{h}(w)$ for $w \in \widehat{\ell}(w)$.
With some additional effort, these facts can be leveraged to show that in this setting $\widehat{h}$ is \emph{continuous}, in the sense that it commutes with (possibly infinite) suprema in $\mathbb{E}$ w.r.t.\ $\leqq$.
\end{remark}

For our purposes, these lemmas tell us that equations of the language model must also hold in any star-continuous KA\@.
This can be seen as a variation of~\cite[Theorem~3]{kozen-smith-1996} for KA \emph{without tests}.

\begin{lemma}%
\label{lemma:lang-vs-cont}
Let $e, f \in \mathbb{E}$.
If $\mathcal{L} \models e = f$, then $\mathfrak{C} \models e = f$.
\end{lemma}
\begin{proof}
Let $K$ be a star-continuous KA, and let $h: \Sigma \to K$.
If $w \in \widehat{\ell}(e) = \widehat{\ell}(f)$, then $w \leqq f$ by \Cref{lemma:embed-sem}, whence $\widehat{h}(w) \leq \widehat{h}(f)$ for all $w \in \widehat{\ell}(e)$.
By \Cref{lemma:cont}, the latter implies that $\widehat{h}(e) \leq \widehat{h}(f)$; similarly, $\widehat{h}(f) \leq \widehat{h}(e)$, and so $\widehat{h}(e) = \widehat{h}(f)$.
\end{proof}

Furthermore, equivalence in all finite relational models implies equivalence in the language model.
We show this by constructing a sequence of finite relational KAs that approximate the language model, in that each of these encapsulates agreement of the language semantics on words up to a certain length.

\begin{lemma}%
\label{lemma:fin-rel-vs-lang}
Let $e, f \in \mathbb{E}$.
If $\mathfrak{FR} \models e = f$, then $\mathcal{L} \models e = f$.
\end{lemma}
\begin{proof}
Our strategy is as follows: given $w \in \Sigma^*$ of length $n$, we design a finite KA $K_n$ and a map $h_w: \Sigma \to K_n$ such that for any $g \in \mathbb{E}$, $\widehat{h_w}(g)$ tells us whether $w \in \widehat{\ell}(g)$.
Given that $e$ and $f$ must be equivalent in any finite KA, it must be the case that $\widehat{h_w}(e) = \widehat{h_w}(f)$, from which we conclude language equivalence.

In more detail, we proceed as follows.
For $n \in \mathbb{N}$, let $K_n = \mathcal{R}(\{ 0, \dots, n \})$.
For $w = \ltr{a}_0 \dots \ltr{a}_{n-1} \in \Sigma^*$, we define $h_w: \Sigma \to K_n$ by $h_w(\ltr{a}) = \{ (i, i+1) : \ltr{a} = \ltr{a}_i \}$.
We now claim the following: for $g \in \mathbb{E}$, we have that $(0, n) \in \widehat{h_w}(g)$ if and only if $w \in \widehat{\ell}(g)$, or more generally, for all $0 \leq i \leq j \leq n$ it holds that $(i, j) \in \widehat{h_w}(g)$ if and only if $\ltr{a}_i \ltr{a}_{i+1} \cdots \ltr{a}_{j-1} \in \widehat{\ell}(g)$.
Suppose this property is true: if $w \in \widehat{\ell}(e)$ with length $n$, then $(0, n) \in \widehat{h_w}(e)$, and because $\mathfrak{FR} \models e = f$, we have that $(0, n) \in \widehat{h_w}(f)$, which tells us that $w \in \widehat{\ell}(f)$.
This shows that $\widehat{\ell}(e) \subseteq \widehat{\ell}(f)$; the proof of the converse inclusion is analogous.

It remains to prove the claim, which we do by induction on $g$.
In the base, there are three cases.
If $g = 0$, then the claim holds vacuously, as $\widehat{h_w}(g) = \widehat{\ell}(g) = \emptyset$.
If $g = 1$, then $(i, j) \in \widehat{h_w}(g)$ if and only if $i = j$, which is equivalent to $\ltr{a}_i\ltr{a}_{i+1}\cdots\ltr{a}_{j-1} = \epsilon$; the latter holds precisely when $\ltr{a}_i\ltr{a}_{i+1}\cdots\ltr{a}_{j-1} \in \widehat{\ell}(g)$.
If $g = \ltr{a}$ for some $\ltr{a} \in \Sigma$, then the claim holds by definition of $h_w$ and $\ell$.

For the inductive step, there are three more cases:
\begin{itemize}
    \item
    If $g = g_1 + g_2$, then we derive
    \begin{align*}
    (i, j) \in \widehat{h_w}(g)
        &\iff (i, j) \in \widehat{h_w}(g_1) \vee (i, j) \in \widehat{h_w}(g_2) \\
        &\iff \ltr{a}_i\ltr{a}_{i+1}\cdots\ltr{a}_{j-1} \in \widehat{\ell}(g_1) \vee \ltr{a}_i\ltr{a}_{i+1}\cdots\ltr{a}_{j-1} \in \widehat{\ell}(g_2) \tag{IH} \\
        &\iff \ltr{a}_i\ltr{a}_{i+1}\cdots\ltr{a}_{j-1} \in \widehat{\ell}(g)
    \end{align*}
    \item
    If $g = g_1 \cdot g_2$, then we derive
    \begin{align*}
    (i, j) \in \widehat{h_w}(g)
        &\iff \exists k.\ (i, k) \in \widehat{h_w}(g_1) \wedge (k, j) \in \widehat{h_w}(g_2) \\
        &\iff \exists k.\ \ltr{a}_i\ltr{a}_{i+1}\cdots\ltr{a}_{k-1} \in \widehat{\ell}(g_1) \wedge \ltr{a}_k\ltr{a}_{k+1}\cdots\ltr{a}_{j-1} \in \widehat{\ell}(g_2) \tag{IH} \\
        &\iff \ltr{a}_i\ltr{a}_{i+1}\cdots\ltr{a}_{j-1} \in \widehat{\ell}(g)
    \end{align*}
    \item
    If $g = g_1^*$, then $(i, j) \in \widehat{h_w}(g) = \widehat{h_w}(g_1)^*$ if and only if $(i, j) \in \widehat{h_w}(g_1)^n$ for some $n \in \mathbb{N}$.
    The latter is equivalent to saying that $\ltr{a}_i\ltr{a}_{i+1}\cdots\ltr{a}_{j-1} \in \widehat{\ell}(g_1)^n$ for some $n \in \mathbb{N}$, which can be shown by induction on $n$ using an argument similar to the previous case.
    Finally, $\ltr{a}_i\ltr{a}_{i+1}\cdots\ltr{a}_{j-1} \in \widehat{\ell}(g_1)^n$ holds for some $n \in \mathbb{N}$ precisely when $\ltr{a}_i\ltr{a}_{i+1}\cdots\ltr{a}_{j-1} \in \widehat{\ell}(g_1)^* = \widehat{\ell}(g)$.

    \vspace{-5mm}
\end{itemize}
\end{proof}

\noindent
These lemmas lead to the implications in \Cref{figure:overview}.
Thus, completeness w.r.t.\ any of $\mathfrak{R}$, $\mathfrak{FR}$, $\mathcal{L}$, $\mathfrak{C}$ or $\mathfrak{F}$ implies completeness w.r.t.\ the others --- e.g., if $\mathcal{L} \models e = f$ implies $e \equiv f$, then so do $\mathfrak{FR} \models e = f$, $\mathfrak{R} \models e = f$, $\mathfrak{F} \models e = f$ and $\mathfrak{C} \models e = f$.
\hyperref[theorem:language-completeness]{Kozen's theorem} then directly leads to the following observation, which we feel the need to record because it seems absent from the literature.

\begin{corollary}[Finite relational model property]
Let $e, f \in \mathfrak{E}$.
If $\mathfrak{FR} \models e = f$, then $e \equiv f$.
\end{corollary}

In the remainder, we work towards a novel proof of \hyperref[theorem:finite-model-property]{Palka's theorem}, which is elementary in the sense that it does not rely on \hyperref[theorem:language-completeness]{Kozen's theorem} nor minimality or bisimilarity of automata.

\begin{figure}
    \centering
    \begin{tikzpicture}
        \node (L) {$\mathcal{L} \models e = f$};
        \node[above=7mm of L] (R) {$\mathfrak{R} \models e = f$};
        \node[left=2cm of L] (FR) {$\mathfrak{FR} \models e = f$};
        \node[right=2cm of L] (C) {$\mathfrak{C} \models e = f$};
        \node[below=7mm of L] (F) {$\mathfrak{F} \models e = f$};

        \draw[-implies,double equal sign distance,rounded corners,gray] (R.west) -| (FR.north);
        \draw[-implies,double equal sign distance,rounded corners,gray] (C.north) |- (R.east);
        \draw[-implies,double equal sign distance,rounded corners,gray] (F.west) -| (FR.south);
        \draw[-implies,double equal sign distance,rounded corners,gray] (C.south) |- (F.east);
        \draw (FR.east) edge[-implies,double equal sign distance] node[above=1mm] {\Cref{lemma:fin-rel-vs-lang}} (L.west);
        \draw (L.east) edge[-implies,double equal sign distance] node[above=1mm] {\Cref{lemma:lang-vs-cont}} (C.west);
    \end{tikzpicture}
    \caption{%
        Overview of the elementary connections between different (classes of) KAs in terms of the equalities that are satisfied.
        The grayed out implications follow from containments between classes of models --- e.g., if $\mathfrak{C} \models e = f$ then $\mathfrak{R} \models e = f$ because all relational models are star-continuous.
    }%
    \label{figure:overview}
\end{figure}
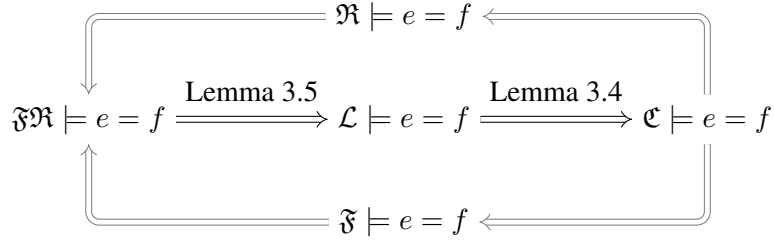

\section{Automata}%
\label{section:solutions-automata}

We start our journey by recalling \emph{automata} as a way of representing languages.
Each automaton can be seen as a system of equations, which induces a \emph{least solution} up to provable equivalence.
None of the material in this section is novel, and in fact, most of it is well-known.
However, least solutions play a central role in the arguments to come, and so it helps to define and discuss them in detail.

\begin{definition}[Automaton]
A (\emph{non-deterministic finite}) \emph{automaton} $A$ is a tuple $(Q, \delta, I, F)$ where $Q$ is a finite set of \emph{states}, $\delta: Q \times \Sigma \to \mathcal{P}(Q)$ is the \emph{transition function} and $I, F \subseteq Q$ contain \emph{initial} and \emph{final states} respectively.

For $\ltr{a} \in \Sigma$, we write $\delta_{\ltr{a}}$ for the relation $\{ (q, q') : q' \in \delta(q, \ltr{a}) \}$.
We can generalize this to words:
\begin{mathpar}
    \delta_\epsilon = \id_Q
    \and
    \delta_{w\ltr{a}} = \delta_w \circ \delta_\ltr{a}
\end{mathpar}

The \emph{language} of $q \in Q$, denoted $L(A, q)$, holds words $w$ such that $\delta_w$ relates $q$ to a final state, i.e., $L(A, q) = \{ w \in \Sigma^* : q \mathrel{\delta_w} q_f \in F \}$.
The language of $A$ is given by $L(A) = \bigcup_{q_i \in I} L(A, q_i)$.
\end{definition}

Automata are commonly drawn as graphs, with labeled vertices representing states and labeled edges representing transitions; initial states are indicated by unlabeled edges without an origin, and accepting states have a double border.
An example is the automaton $A_\mathsf{alt}$ depicted in \Cref{figure:automaton-alt}.

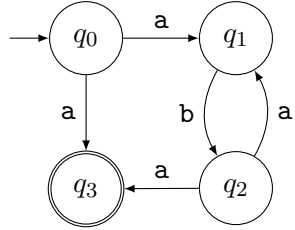
\begin{figure}
    \centering
    \begin{tikzpicture}[every initial by arrow/.style={-latex}]
        \node[state,initial,initial text={}] (q0) {$q_0$};
        \node[state,right=of q0] (q1) {$q_1$};
        \node[state,below=of q1] (q2) {$q_2$};
        \node[state,accepting,below=of q0] (q3) {$q_3$};
        \draw (q0) edge[-latex] node[above] {$\ltr{a}$} (q1);
        \draw (q1) edge[-latex,bend right] node[left] {$\ltr{b}$} (q2);
        \draw (q2) edge[-latex,bend right] node[right] {$\ltr{a}$} (q1);
        \draw (q2) edge[-latex] node[above] {$\ltr{a}$} (q3);
        \draw (q0) edge[-latex] node[left] {$\ltr{a}$} (q3);
    \end{tikzpicture}
    \caption{Visual representation of an automaton $A_\mathsf{alt}$ accepting the language ${(\ltr{a}\cdot\ltr{b})}^*\cdot\ltr{a}$.}%
    \label{figure:automaton-alt}
\end{figure}

Kleene's theorem says that the set of languages defined by regular expressions is the same as the set of languages described by automata~\cite{kleene-1956}.
In fact, the translations that demonstrate this equivalence play an important role going forward; we will now outline the translation from automata to expressions.

For the next section and the remainder of this one, we fix an automaton $A = (Q, \delta, I, F)$.

\begin{definition}[Solutions]
Let $e \in \mathbb{E}$.
An \emph{$e$-solution} to $A$ is a function $s: Q \to \mathbb{E}$ such that the following hold for all $q, q' \in Q$:
\begin{mathpar}
    q \in F \implies {e \leqq s(q)}
    \and
    q' \in \delta(q, \mathtt{a}) \implies {\mathtt{a} \cdot s(q') \leqq s(q)}
\end{mathpar}
A $1$-solution to $A$ is simply called a \emph{solution} to $A$; alternatively, if $s$ is a solution to $A$, we may also write that $s$ \emph{solves} $A$.
Finally, we say that $s$ is the \emph{least $e$-solution} (resp.\ \emph{least solution}) to $A$ if for all $e$-solutions (resp.\ solutions) $s'$ and for all $q \in Q$ it holds that $s(q) \leqq s'(q)$.
\end{definition}

It is not too hard to show that if $s$ is the least solution to $A$, then $\widehat{\ell}(s(q)) = L(A, q)$ for all states $q \in Q$; thus, the ability to compute least solutions to an automaton amounts to a conversion from automata to equivalent expressions, which constitutes one half of Kleene's theorem.

\begin{example}
Consider the automaton $A_\mathsf{alt}$ from \Cref{figure:automaton-alt}.
A solution to this automaton is a function $s: \{ q_0, q_1, q_2, q_3 \} \to \mathbb{E}$ satisfying all of the following conditions:
\begin{mathpar}
    1 \leqq s(q_3)
    \and
    \ltr{a} \cdot s(q_3) \leqq s(q_0)
    \and
    \ltr{a} \cdot s(q_1) \leqq s(q_0)
    \\
    \ltr{a} \cdot s(q_3) \leqq s(q_2)
    \and
    \ltr{a} \cdot s(q_1) \leqq s(q_2)
    \and
    \ltr{b} \cdot s(q_2) \leqq s(q_1)
\end{mathpar}
To satisfy all of these conditions, we could choose the constant function $s_1$ that assigns ${(\ltr{a}+\ltr{b})}^*$ to each state.
Alternatively, we could choose $s_2(q_0) = s_2(q_2) = {(\ltr{a}\cdot\ltr{b})}^* \cdot\ltr{a}$, $s(q_1) = \ltr{b}\cdot{(\ltr{a}\cdot\ltr{b})}^*\cdot\ltr{a}$ and $s_2(q_3) = 1$.
One can show that $s_2(q) \leqq s_1(q)$ for all states $q$; indeed, $s_2$ is the least solution.
\end{example}

Least $e$-solutions are unique up to the laws of KA --- i.e., if $s$ and $s'$ are both least $e$-solutions to $A$, then clearly $s(q) \equiv s'(q)$ for all $q \in Q$.
This is why we speak of \emph{the} least $e$-solution (resp.\ least solution) to an automaton.
As it turns out, the least $e$-solution to an automaton always exists, for each $e$: the process to compute them~\cite{conway-1971,kozen-1994} arises quite naturally from the \emph{state elimination} algorithm~\cite{kleene-1956,mcnaughton-yamada-1960}, which produces a regular expression that represents the language accepted by a given automaton.

Readers familiar with the state elimination algorithm will recall that, as an intermediate form, it requires automata that can have transitions labeled by regular expressions, rather than just letters.
On the algebraic side, a similar generalization is necessary~\cite{conway-1971,kozen-1994}, which we detail below.

\begin{definition}
Let $Q$ be finite.
A \emph{$Q$-vector} is a function $b: Q \to \mathbb{E}$, and a \emph{$Q$-matrix} is a function $M: Q \times Q \to \mathbb{E}$.
Given a $Q$-vector $M$ and a $Q$-vector $b$, we refer to the pair $(M, b)$ as a \emph{linear system (on $Q$)}.
Let $e \in \mathbb{E}$; we call $s: Q \to \mathbb{E}$ an \emph{$e$-solution} to $(M, b)$ when for all $q, q' \in Q$, we have:
\begin{mathpar}
    b(q) \cdot e \leqq s(q)
    \and
    M(q, q') \cdot s(q') \leqq s(q)
\end{mathpar}
We expand the terminology from earlier: if $s: Q \to \mathbb{E}$ is an $e$-solution to $(M, b)$ such that for all $e$-solutions $s': Q \to \mathbb{E}$ to $(M, b)$ and all $q \in Q$ we have $s(q) \leqq s'(q)$, then $s$ is the \emph{least} $e$-solution; a $1$-solution to $(M, b)$ is simply called a solution, and least solutions are defined similarly.
\end{definition}

Linear systems always admit a least $e$-solution, and it is this $e$-solution that will allow us to compute the least $e$-solution to an automaton momentarily.
The next theorem can be viewed as a different perspective on the proof by Conway~\cite[Chapter~3]{conway-1971}, who showed that matrices over what he called an $\mathbf{S}$-algebra again form an $\mathbf{S}$-algebra, which was applied to Kleene algebra by Kozen~\cite{kozen-1994}.
Conway and Kozen used a divide-and-conquer approach to compute the Kleene star of a square matrix; the proof below is closer to the one by Backhouse~\cite{backhouse-1975}, who proceeded by induction on the size of the matrix.

In the statement that follows and throughout the sequel, given a $Q$-vector $s$ and $e \in \mathbb{E}$, we will write $s \cdot e$ for the $Q$-vector given by $(s \cdot e)(q) = s(q) \cdot e$ --- i.e., the scalar multiplication of $s$ by $e$.

\begin{theorem}%
\label{theorem:computing-matrix-solutions}
Let $Q$ be finite.
Given $M: Q \times Q \to \mathbb{E}$ and $b: Q \to \mathbb{E}$, we can compute a $Q$-indexed vector $s: Q \to \mathbb{E}$ such that for all $e \in \mathbb{E}$, we have that $s \cdot e$ is the least $e$-solution to $(M, b)$.
\end{theorem}
\begin{proof}
We proceed by induction on $Q$.
In the base, where $Q = \emptyset$, the claim holds vacuously.
For the inductive step, fix any $q \in Q$, and let $P = Q \setminus \{ q \}$.
We now choose $N: P \times P \to \mathbb{E}$ and $c: P \to \mathbb{E}$ as follows:
\begin{mathpar}
    N(p, p') = M(p, p') + M(p, q) \cdot {M(q, q)}^* \cdot M(q, p')
    \and
    c(p) = b(p) + M(p, q) \cdot {M(q, q)}^* \cdot b(q)
\end{mathpar}
By induction we can now compute a $t: P \to \mathbb{E}$ such that for all $e \in \mathbb{E}$, we have that $t \cdot e$ is the least $e$-solution to $(N, c)$.
We base our choice for $s$ on $t$, as follows:
\[
    s(q') =
        \begin{cases}
        {M(q, q)}^* \cdot (b(q) + \sum_{p \in P} M(q, p) \cdot t(p)) & q' = q \\
        t(q') & q' \neq q
        \end{cases}
\]
We now need to verify that $s$ satisfies the properties claimed.
To see that $s \cdot e$ is an $e$-solution to $(M, b)$, it suffices to show that $s$ is a solution.
For the first constraint, note that if $q' = q$, then $b(q') \leqq {M(q,q)}^* \cdot b(q) \leqq s(q) = s(q')$ by definition of $s$, and if $q' \neq q$ then $b(q') \leqq c(q') \leqq t(q') = s(q')$, by the fact that $t$ is a solution to $(N, c)$.
For the second constraint, let $q', q'' \in Q$; there are four cases.
\begin{itemize}
    \item
    If $q' = q = q''$ then by definition of $s$ and the fact that $M(q,q) \cdot {M(q,q)}^* \leqq {M(q,q)}^*$, we have:
    \begin{align*}
        M(q, q) \cdot s(q)
            &= M(q, q) \cdot {M(q, q)}^* \cdot (b(q) + \sum_{p \in P} M(q, p) \cdot t(p)) \\
            &\leqq {M(q, q)}^* \cdot (b(q) + \sum_{p \in P} M(q, p) \cdot t(q))
             = s(q)
    \end{align*}
    \item
    If $q' \neq q = q''$ then by definition of $s$, $N$ and $c$ as well as the fact that $t$ solves $(N, c)$, we find:
    \begin{align*}
        M(q', q) \cdot s(q)
            &= M(q', q) \cdot {M(q, q)}^* \cdot (b(q) + \sum_{p \in P} M(q, p) \cdot t(p)) \\
            &\equiv M(q', q) \cdot {M(q, q)}^* \cdot b(q) + \sum_{p \in P} M(q', q) \cdot {M(q, q)}^* \cdot M(q, p) \cdot t(q) \\
            &\leqq c(q') + \sum_{p \in P} N(q', p) \cdot t(p)
             \leqq t(q') = s(q')
    \end{align*}
    \item
    If $q' = q \neq q''$ then by definition of $s$ we can derive that:
    \begin{align*}
        M(q, q'') \cdot s(q'')
            &= M(q, q'') \cdot t(q'')
             \leqq {M(q,q)}^* \cdot M(q, q'') \cdot t(q'')
             \leqq s(q)
    \end{align*}
    \item
    If $q' \neq q \neq q''$ then by definition of $s$ and $N$, as well as the fact that $t$ solves $(N, c)$, we have:
    \begin{align*}
        M(q', q'') \cdot s(q'')
            &\leqq N(q', q'') \cdot t(q'')
             \leqq t(q') = s(q')
    \end{align*}
\end{itemize}
To show that $s \cdot e$ is the least $e$-solution to $(M, b)$, let $s': Q \to \mathbb{E}$ be any $e$-solution to $(M, b)$; we wish to show that $s(q) \cdot e \leqq s'(q)$ for all $q \in Q$.
To this end, choose $t': P \to \mathbb{E}$ by setting $t'(p) = s'(p)$ for all $p \in P$.
We claim that $t'$ is an $e$-solution to $(N, c)$.
For the first constraint, we first observe that $M(q, q) \cdot s'(q) + b(q) \cdot e \leqq s'(q)$ because $s'$ is an $e$-solution to $(M, b)$, and so ${M(q, q)}^* \cdot b(q) \cdot e \leqq s'(q)$ by the fixpoint axiom.
With this fact in hand, we can then verify the first constraint as follows:
\begin{align*}
    c(p) \cdot e
        &\equiv b(p) \cdot e + M(p, q) \cdot {M(q, q)}^* \cdot b(q) \cdot e \\
        &\leqq b(p) \cdot e + M(p, q) \cdot s'(q) \leqq s'(p) = t'(p)
\end{align*}
For the second constraint, we note that $M(q, q) \cdot s'(q) + s'(q) \leqq s'(q)$, and so $M(q, q)^* \cdot s'(q) \leqq s'(q)$ by the fixpoint rule.
We apply this property and the fact that $s'$ is an $e$-solution to $(M, b)$ to derive:
\begin{align*}
    N(p, p') \cdot t'(p')
        &\equiv M(p, p') \cdot s'(p) + M(p, q) \cdot {M(q,q)}^* \cdot M(q, p') \cdot s'(p') \\
        &\leqq M(p, p') \cdot s'(p) + M(p, q) \cdot {M(q,q)}^* \cdot s'(q) \\
        &\leqq M(p, p') \cdot s'(p) + M(p, q) \cdot s'(q)
         \leqq s'(p)
         = t'(p)
\end{align*}
This makes $t'$ an $e$-solution to $(N, c)$, and thus by induction we have for $p \in P$ that $s(p) \cdot e = t(p) \cdot e \leqq t'(p) = s'(p)$.
It remains to show $s(q) \cdot e \leqq s'(q)$.
Since $t(p) \cdot e \leqq t'(p)$ for all $p \in P$, we find:
\[
    M(q, q) \cdot s'(q) + (b(q) + \sum_{p \in P} M(q, p) \cdot t(p)) \cdot e
        \leqq s'(q)
\]
and so $s(q) \cdot e = {M(q, q)}^* \cdot (b(q) + \sum_{p \in P} M(q, p) \cdot t(p)) \cdot e \leqq s'(q)$ by the fixpoint rule.
\end{proof}

Least $e$-solutions to linear systems then readily allow us to compute least $e$-solutions to automata, by casting the automaton in the form of a linear system and arguing that they have the same solutions.

\begin{theorem}[Solutions via state elimination]%
\label{theorem:computing-solutions}
We can compute a vector $\sol{A}: Q \to \mathbb{E}$, such that for all $e \in \mathbb{E}$, $\sol{A} \cdot e$ is the least $e$-solution to $A$.
\end{theorem}
\begin{proof}
We choose a $Q$-indexed vector $b$ and a $Q$-indexed matrix $M$ as follows:
\begin{mathpar}
    b(q) =
        \begin{cases}
        1 & q \in F \\
        0 & q \not\in F
        \end{cases}
    \and
    M(q, q') =
        \sum_{q' \in \delta(q, \mathtt{a})} \mathtt{a}
\end{mathpar}
We now claim that $s: Q \to \mathbb{E}$ is an $e$-solution to $(M, b)$ iff it is an $e$-solution to $A$.

For the forward claim, let $s$ be an $e$-solution to $(M, b)$.
If $q \in F$, then $e \equiv b(q) \cdot e \leqq s(q)$; also, if $q' \in \delta(q, \ltr{a})$, then $\ltr{a} \cdot s(q) \leqq M(q, q') \cdot s(q') \leqq s(q)$ --- hence, $s$ is an $e$-solution to $A$.

Conversely, let $s$ be an $e$-solution to $A$.
If $q \in F$, then $b(q) \cdot e \equiv e \leqq s(q)$; otherwise, if $q \not\in F$, then $b(q) \cdot e \equiv 0 \leqq s(q)$.
Moreover, if $q, q' \in Q$, then $\ltr{a} \cdot s(q') \leqq s(q)$ when $q' \in \delta(q, \ltr{a})$, so we find:
\[
    M(q, q') \cdot s(q')
        \equiv \Bigl( \sum_{q' \in \delta(q, \ltr{a})} \ltr{a} \Bigr) \cdot s(q')
        \equiv \sum_{q' \in \delta(q, \ltr{a})} \ltr{a} \cdot s(q')
        \leqq s(q)
\]

By \Cref{theorem:computing-matrix-solutions}, we know that we can compute a vector $s: Q \to \mathbb{E}$ such that for all $e \in \mathbb{E}$, we have that $s \cdot e$ is the least $e$-solution to $(M, b)$; by the above, we can choose $\sol{A} = s$ to satisfy the claim: $\sol{A}$ is a solution to $(M, b)$ which makes it a solution to $A$, and if $s'$ is an $e$-solution to $A$ then it is also an $e$-solution to $(M, b)$, which tells us that $\sol{A}(q) \cdot e \leqq s'(q)$ for all $q \in Q$.
\end{proof}

\Cref{theorem:computing-solutions} forms the technical nexus of our arguments, and we apply it repeatedly in the next three sections.
We write $\soli{A}$ for the sum of least solutions to initial states in $A$, i.e., $\sum_{q \in I} \sol{A}(q)$.

\section{Transformation automata}%
\label{section:transformation-automata}

We now turn our attention to \emph{transformation automata}~\cite{mcnaughton-papert-1968}.
The reachable states of a transformation automaton $A'$ obtained from $A$ are relations on $Q$, representing transformations that can be effected by reading a word in $A$.
By varying the accepting states in the transformation automaton for $A$, we can obtain regular expressions for the sets of words that transform the states of $A$ in a particular way.

\begin{definition}
We define $\delta^\tau: \mathcal{R}(Q) \times \Sigma \to \mathcal{P}(\mathcal{R}(Q))$ by setting $\delta^\tau(R, \ltr{a}) = \{ R \circ \delta_{\ltr{a}} \}$.
For each $R \in \mathcal{R}(Q)$, write $A[R]$ for the \emph{$R$-transformation automaton} $(\mathcal{R}(Q), \delta^\tau, \{ \id \}, \{ R \})$.
\end{definition}
Note that the above still fits our definition of an automaton, since if $Q$ is finite then so is $\mathcal{R}(Q)$.

\begin{example}
As an example of a transformation automaton, consider \Cref{figure:transformation-automaton-alt}, where we depict the reachable part of the transformation automaton obtained from $A_\mathsf{alt}$ (\Cref{figure:automaton-alt}), but without any accepting state selected.
Note how each state in the transformation automaton is a relation on states of the original automaton, represented by $\delta_w$ for some word $w$, and that some words bring about the same transformation on states as others.
For instance, reading the word $\ltr{a}\ltr{a}$ from any state in $A_\mathsf{alt}$ will not lead to another state, which also happens when we read the word $\ltr{b}\ltr{b}$ from any state.
\end{example}

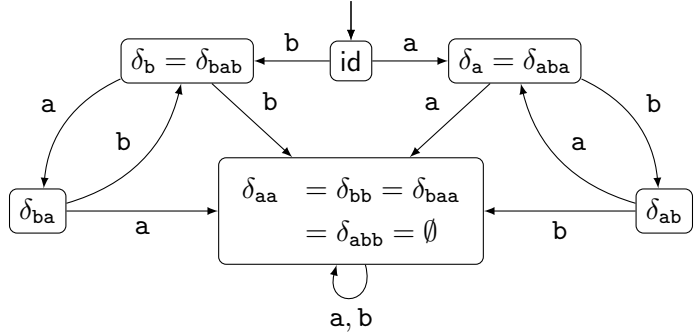
\begin{figure}
    \centering
    \begin{tikzpicture}[every initial by arrow/.style={-latex}]
        \node[draw,rounded corners=1mm,initial,initial above,initial text={}] (id) {$\id$};
        \node[draw,rounded corners=1mm,below=of id] (empty) {$\begin{array}{rl}\delta_\ltr{aa} &= \delta_\ltr{bb} = \delta_\ltr{baa} \\ &= \delta_\ltr{abb} = \emptyset\end{array}$};
        \node[draw,rounded corners=1mm,right=of id] (da) {$\delta_\ltr{a} = \delta_\ltr{aba}$};
        \node[draw,rounded corners=1mm,right=2cm of empty] (dab) {$\delta_\ltr{ab}$};
        \node[draw,rounded corners=1mm,left=of id] (db) {$\delta_\ltr{b} = \delta_\ltr{bab}$};
        \node[draw,rounded corners=1mm,left=2cm of empty] (dba) {$\delta_\ltr{ba}$};
        \draw (id) edge[-latex] node[above] {$\ltr{a}$} (da);
        \draw (id) edge[-latex] node[above] {$\ltr{b}$} (db);
        \draw (db) edge[-latex,bend right] node[above left] {$\ltr{a}$} (dba);
        \draw (dba) edge[-latex,bend right] node[above left] {$\ltr{b}$} (db);
        \draw (da) edge[-latex,bend left] node[above right] {$\ltr{b}$} (dab);
        \draw (dab) edge[-latex,bend left] node[above right] {$\ltr{a}$} (da);
        \draw (dba) edge[-latex] node[below] {$\ltr{a}$} (empty);
        \draw (dab) edge[-latex] node[below] {$\ltr{b}$} (empty);
        \draw (db) edge[-latex] node[above right] {$\ltr{b}$} (empty);
        \draw (da) edge[-latex] node[above left] {$\ltr{a}$} (empty);
        \draw (empty) edge[-latex,loop below,looseness=4] node[below] {$\ltr{a}, \ltr{b}$} (empty);
    \end{tikzpicture}
    \caption{%
        Template for $A_\mathsf{alt}[-]$, the transformation automaton of $A_\mathsf{alt}$ (\Cref{figure:automaton-alt}).
        For the sake of simplicity, no accepting state (relation) has been selected, and only the reachable part of the automaton has been drawn.
    }%
    \label{figure:transformation-automaton-alt}
\end{figure}

Readers familiar with formal language theory will recognize transformation automata as the link between regular and recognizable languages~\cite{mcnaughton-papert-1968}.
The novelty of this section therefore lies in the connection between the least solution of an automaton and its transformation automata, as recorded in the lemmas that follow.
We start by analyzing the solutions to transformation automata.
A useful first observation is that, for each $\ltr{a} \in \Sigma$, words read from $\id$ to $\delta_{\ltr{a}}$ in the transformation automaton include $\ltr{a}$.

\begin{lemma}%
\label{lemma:solve-transformation-aut-letter}
For all $\ltr{a} \in \Sigma$, it holds that $\ltr{a} \leqq \soli{A[\delta_{\ltr{a}}]}$.
\end{lemma}
\begin{proof}
Using the fact that $\delta_{\ltr{a}} \in \delta^\tau(\id, \ltr{a})$ and $\sol{A[\delta_{\ltr{a}}]}$ solves $A[\delta_{\ltr{a}}]$ (by \Cref{theorem:computing-solutions}), we derive as follows:
\[
    \ltr{a}
        \equiv \ltr{a} \cdot 1
        \leqq \ltr{a} \cdot \sol{A[\delta_{\ltr{a}}]}(\delta_{\ltr{a}})
        \leqq \sol{A[\delta_{\ltr{a}}]}(\id)
        \equiv \soli{A[\delta_{\ltr{a}}]}
\]

\vspace{-9mm}
\end{proof}

Transformation automata also enjoy the following property: if we can read $w$ while moving from $R_1$ to $R_2$, then we can also read $w$ while moving from $R_3 \circ R_1$ to $R_3 \circ R_2$.
For example, in the transformation automaton from \Cref{figure:transformation-automaton-alt}, we can read $\ltr{a}\ltr{b}\ltr{a}$ while moving from $\delta_\ltr{ab}$ to $\delta_\ltr{a}$, which implies that we can also read $\ltr{a}\ltr{b}\ltr{a}$ while moving from $\delta_\ltr{b} \circ \delta_\ltr{ab} = \delta_\ltr{bab} = \delta_\ltr{b}$ to $\delta_\ltr{b} \circ \delta_\ltr{a} = \delta_\ltr{ba}$, as is indeed the case.
Because solutions to automata encode the expressions that describe the paths to an accepting state, and because the way we set up transformation automata lets us choose their accepting states, this property can also be expressed in terms of solutions to transformation automata.

\begin{lemma}%
\label{lemma:solve-transformation-aut-shift}
For all $R_1, R_2, R_3 \subseteq Q \times Q$, it holds that $\sol{A[R_2]}(R_1) \leqq \sol{A[R_3 \circ R_2]}(R_3 \circ R_1)$.
\end{lemma}
\begin{proof}
We fix $R_2$ and $R_3$, and choose $s: \mathcal{R}(Q) \to \mathbb{E}$ by setting $s(R) = \sol{A[R_3 \circ R_2]}(R_3 \circ R)$.
By \Cref{theorem:computing-solutions}, it suffices to prove that $s$ is a solution to $A[R_2]$.
First, if $R$ is accepting in $A[R_2]$, then $R = R_2$, in which case $1 \leqq \sol{A[R_3 \circ R_2]}(R_3 \circ R_2) = s(R_2)$ because $\sol{A[R_3 \circ R_2]}$ is a solution to $A[R_3 \circ R_2]$.

Next, if $R' \in \delta^\tau(R, \ltr{a})$, then $R' = R \circ \delta_{\ltr{a}}$, which means that $R_3 \circ R' = R_3 \circ R \circ \delta_{\ltr{a}}$, and therefore $R_3 \circ R' \in \delta^\tau(R_3 \circ R, \ltr{a})$.
Because $\sol{A[R_3 \circ R_2]}$ is a solution to $A[R_3 \circ R_2]$, we can then derive that
\[
    \ltr{a} \cdot s(R')
        \equiv \ltr{a} \cdot \sol{A[R_3 \circ R_2]}(R_3 \circ R')
        \leqq \sol{A[R_3 \circ R_2]}(R_3 \circ R) = s(R)
\]

\vspace{-8mm}
\end{proof}

We can think of $\soli{A[R]}$ as an expression representing words $w$ such that $\delta_w = R$, i.e., the set of words that transform the state space of $A$ in the same way that $R$ does.
It then stands to reason that a word representing the transformation $R_1$ followed by a word representing the transformation $R_2$ can be concatenated to form a word representing the transformation $R_1 \circ R_2$.
For instance, if we consult \Cref{figure:transformation-automaton-alt}, then we find that $\ltr{a}$ is accepted by $A_\mathsf{alt}[\delta_\ltr{a}]$ and $\ltr{b}$ is accepted by $A_\mathsf{alt}[\delta_\ltr{b}]$, which implies that $\ltr{ab}$ is accepted by $A_\mathsf{alt}[\delta_\ltr{ab}]$.
Thus, the language of $A[R_1]$ concatenated to that of $A[R_2]$ should be contained in the language of $A[R_1 \circ R_2]$.
We prove this using KA in the following lemma.

\begin{lemma}%
\label{lemma:solve-transformation-aut-compose}
For all $R_1, R_2 \subseteq Q \times Q$ it holds that
\(
    \soli{A[R_1]} \cdot \soli{A[R_2]} \leqq \soli{A[R_1 \circ R_2]}
\)
\end{lemma}
\begin{proof}
Because all transformation automata have the same initial state $\id$, it suffices to show that $\sol{A[R_1 \circ R_2]}$ is an $\soli{A[R_2]}$-solution to $A[R_1]$, by \Cref{theorem:computing-solutions}.
First, if $R$ is accepting in $A[R_1]$, then $R = R_1$.
In that case, we find that $\soli{A[R_2]} = \sol{A[R_2]}(\id) \leqq \sol{A[R_1 \circ R_2]}(R_1) = \sol{A[R_1 \circ R_2]}(R)$ by \Cref{lemma:solve-transformation-aut-shift}.

Next, if $R' \in \delta^\tau(R, \ltr{a})$, then we need to show that $\ltr{a} \cdot \sol{A[R_1 \circ R_2]}(R') \leqq \sol{A[R_1 \circ R_2]}(R)$, but this is immediately true because $\sol{A[R_1 \circ R_2]}$ solves $A[R_1 \circ R_2]$ (also by \Cref{theorem:computing-solutions}).
\end{proof}

The final technical property that we need expresses that if $R$ relates $q$ to an accepting state $q_f \in F$, and $w$ is accepted by $A[R]$, then $w$ is accepted by $q$ in $A$.
This too can be explained using the intuition behind transformation automata, because in this situation $R$ is the transformation in $A$ brought about by reading $w$, which apparently takes $q$ to an accepting state $q_f$.
For instance, consider the transformation automaton in \Cref{figure:transformation-automaton-alt}.
If we choose $\delta_\ltr{ba}$ as accepting state, then this automaton accepts $\ltr{b}\ltr{a}$.
Furthermore, $\delta_\ltr{ba}$ relates $q_1$ to $q_3$, which means that $q_1$ accepts $\ltr{b}\ltr{a}$.
We can express this on the level of KA as well.

\begin{lemma}%
\label{lemma:solve-transformation-aut-approximate-below}
If $q \in Q$, $q_f \in F$ and $R \subseteq Q \times Q$ such that $q \mathrel{R} q_f$, then $\soli{A[R]} \leqq \sol{A}(q)$.
\end{lemma}
\begin{proof}
We define $s: \mathcal{R}(Q) \to \mathbb{E}$ by choosing for $R' \subseteq Q \times Q$ that $s(R') = \sum_{q \mathrel{R'} q'} \sol{A}(q')$, and claim that $s$ is a solution to $A[R]$.
First, if $R'$ is accepting in $A[R]$, then $R = R'$.
Since $q \mathrel{R} q_f$, we have that $\sol{A}(q_f) \leqq s(R')$.
Because $q_f \in F$, we also know that $1 \leqq \sol{A}(q_f)$; this allows us to conclude $1 \leqq s(R')$.

If $R'' \in \delta^\tau(R', \ltr{a})$, then $R'' = R' \circ \delta_{\ltr{a}}$.
We now derive as follows
\begin{align*}
    \ltr{a} \cdot s(R'')
        &\equiv \sum_{q \mathrel{R''} q''} \ltr{a} \cdot \sol{A}(q'') \tag{def. $s$, distributivity}\\
        &\equiv \sum_{q \mathrel{R'} q'} \sum_{q'' \in \delta(q', \ltr{a})} \ltr{a} \cdot \sol{A}(q'') \tag{$R'' = R' \circ \delta_{\ltr{a}}$} \\ 
        &\leqq \sum_{q \mathrel{R'} q'} \sol{A}(q') = s(R') \tag{$\sol{A}$ solves $A$}
\end{align*}
Since $s$ solves $A[R]$, we can apply \Cref{theorem:computing-solutions} to find that $\soli{A[R]} = \sol{A[R]}(\id) \leqq s(\id) = \sol{A}(q)$.
\end{proof}

\section{Antimirov's construction}%
\label{section:antimirov-construction}

We now discuss automata $A_e$ that accept $\widehat{\ell}(e)$ for each $e \in \mathbb{E}$, and their solutions.
Many methods to obtain such an automaton exist~\cite{thompson-1968,brzozowski-1964} --- we refer to~\cite{watson-1993} for a good overview.
In this section, we recall Antimirov's construction~\cite{antimirov-1996} in our own notation, and argue two crucial properties.

The idea behind Antimirov's construction is to endow expressions with the structure of an automaton, as was done by Brzozowski~\cite{brzozowski-1964}.
The language of a state represented by $e \in \mathbb{E}$ is meant to be $\widehat{\ell}(e)$, i.e., the language of its expression.
From this perspective, the accepting states should be those representing expressions whose language contains the empty word.
These expressions are fairly easy to describe.

\begin{definition}
The set $\mathbb{F}$ is defined as the smallest subset of $\mathbb{E}$ satisfying:
\begin{mathpar}
    \inferrule{~}{%
        1 \in \mathbb{F}
    }
    \and
    \inferrule{%
        e_1 \in \mathbb{F} \\
        e_2 \in \mathbb{E}
    }{%
        e_1 + e_2 \in \mathbb{F}
    }
    \and
    \inferrule{%
        e_1 \in \mathbb{E} \\
        e_2 \in \mathbb{F}
    }{%
        e_1 + e_2 \in \mathbb{F}
    }
    \and
    \inferrule{%
        e_1, e_2 \in \mathbb{F}
    }{%
        e_1 \cdot e_2 \in \mathbb{F}
    }
    \and
    \inferrule{%
        e_1 \in \mathbb{E}
    }{%
        e_1^* \in \mathbb{F}
    }
\end{mathpar}
\end{definition}
A straightforward argument by induction on $e$ tells us that $\epsilon \in \widehat{\ell}(e)$ if and only if $e \in \mathbb{F}$.

For example, $1 \cdot {(\ltr{b} \cdot \ltr{a})}^* \in \mathbb{F}$ because $1 \in \mathbb{F}$ and ${(\ltr{b} \cdot \ltr{a})}^* \in \mathbb{F}$; indeed, $\epsilon \in \widehat{\ell}(1 \cdot {(\ltr{b} \cdot \ltr{a})}^*)$.
Meanwhile, $\ltr{a} \cdot {(\ltr{b} \cdot \ltr{a})}^* \not\in \mathbb{F}$ despite ${(\ltr{b} \cdot \ltr{a})}^* \in \mathbb{F}$, because $\ltr{a} \not\in \mathbb{F}$; this matches the fact that $\epsilon \not\in \ltr{a} \cdot {(\ltr{b} \cdot \ltr{a})}^* \not\in \mathbb{F}$.

\smallskip
Next, we recall Antimirov's transition function.
Intuitively, an expression $e$ has an $\ltr{a}$-transition to an expression $e'$ when $e'$ denotes remainders of words in $\widehat{\ell}(e)$ that start with $\ltr{a}$ --- i.e., if $w \in \widehat{\ell}(e')$, then $\ltr{a}w \in \widehat{\ell}(e)$.
Together, the expressions reachable by $\ltr{a}$-transitions from $e$ describe \emph{all} such words.

In the following, when $S \subseteq \mathbb{E}$ and $e \in \mathbb{E}$, we write $S \fatsemi e$ for $\{ e' \cdot e : e' \in S \}$.
We also write $e \star S$ as a shorthand for $S$ when $e \in \mathbb{F}$, and $\emptyset$ otherwise.

\begin{definition}
We define $\partial: \mathbb{E} \times \Sigma \to \mathcal{P}(\mathbb{E})$ recursively, as follows
\begin{align*}
    \partial(0, \ltr{a}) &= \emptyset
        & \partial(e_1 + e_2, \ltr{a}) &= \partial(e_1, \ltr{a}) \cup \partial(e_2, \ltr{b}) \\
    \partial(1, \ltr{a}) &= \emptyset
        & \partial(e_1 \cdot e_2, \ltr{a}) &= \partial(e_1, \ltr{a}) \fatsemi e_2 \cup e_1 \star \partial(e_2, \ltr{a}) \\
    \partial(\ltr{b}, \ltr{a}) &= \{ 1 : \ltr{a} = \ltr{b} \}
        & \partial(e_1^*, \ltr{a}) &= \partial(e_1, \ltr{a}) \fatsemi e_1^*
\end{align*}
\end{definition}

Antimirov's transition function can decompose an expression $e$ into several ``derivatives'' $e'$, each of which describes the ``tails'' of words denoted by $e$, while $\mathbb{F}$ tells us whether $e$ accepts the empty word.
The following lemma relates this information back to $e$.
It can be viewed as one half of the \emph{fundamental theorem} of Kleene algebra~\cite{rutten-2000}, which says that an expression can be reconstructed from its derivatives, and whether or not it appears in $\mathbb{F}$; we prove this weaker version, because it is all we need.

\begin{lemma}%
\label{lemma:half-fundamental}
Let $e \in \mathbb{E}$.
If $e \in \mathbb{F}$, then $1 \leqq e$.
Furthermore, if $\ltr{a} \in \Sigma$ and $e' \in \partial(e, \ltr{a})$, then $\ltr{a} \cdot e' \leqq e$.
\end{lemma}
\begin{proof}
The first claim follows by a straightforward induction on the construction of $\mathbb{F}$.
The proof of the second claim proceeds by induction on $e$.
In the base, the cases where $e = 0$ or $e = 1$ hold vacuously, as $\partial(e, \ltr{a}) = \emptyset$.
If $e = \ltr{b}$ for some $\ltr{b} \in \Sigma$, then $e' \in \partial(e, \ltr{a})$ implies that $e' = 1$ and $\ltr{a} = \ltr{b}$.
We then have that $\ltr{a} \cdot e' = \ltr{b} \cdot 1 \leqq \ltr{b} = e$.
For the inductive step, there are three more cases.
\begin{itemize}
    \item
    If $e = e_1 + e_2$, then w.l.o.g.\ $e' \in \partial(e_1, \ltr{a})$.
    By induction, we then have that $\ltr{a} \cdot e' \leqq e_1 \leqq e$.
    \item
    If $e = e_1 \cdot e_2$, then there are two subcases.
    First, if $e' = e_1' \cdot e_2$ with $e_1' \in \partial(e_1, \ltr{a})$, then by induction $\ltr{a} \cdot e_1' \leqq e_1$, so $\ltr{a} \cdot e' = \ltr{a} \cdot e_1' \cdot e_2 \leqq e_1 \cdot e_2 = e$.
    Otherwise, if $e_1 \in \mathbb{F}$ and $e' \in \partial(e_2, \ltr{a})$, then $\ltr{a} \cdot e' \leqq e_2$ by induction.
    Furthermore, since $e_1 \in \mathbb{F}$, by the first part of the lemma we have that $1 \leqq e_1$.
    Putting this together, we find that $\ltr{a} \cdot e' \leqq e_2 \equiv 1 \cdot e_2 \leqq e_1 \cdot e_2 = e$.
    \item
    If $e = e_1^*$, then $e' = e_1' \cdot e_1^*$ with $e_1' \in \partial(e_1, \ltr{a})$.
    By induction, we then know that $\ltr{a} \cdot e_1' \leqq e_1$, and hence $\ltr{a} \cdot e' = \ltr{a} \cdot e_1' \cdot e_1^* \leqq  e_1 \cdot e_1^* \leqq e_1^* = e$.
\end{itemize}

\vspace{-8mm}
\end{proof}

The expression $e$ could serve as the sole initial state in the automaton for $e$.
However, our automata allow multiple initial states, and distributing this task among them simplifies our arguments.

\begin{definition}
We define $\iota: \mathbb{E} \to \mathcal{P}(\mathbb{E})$ recursively, as follows:
\begin{align*}
    \iota(0) &= \emptyset
        & \iota(e_1 + e_2) &= \iota(e_1) \cup \iota(e_2) \\
    \iota(1) &= \{ 1 \}
        & \iota(e_1 \cdot e_2) &= \iota(e_1) \fatsemi e_2 \\
    \iota(\ltr{a}) &= \{ \ltr{a} \}
        & \iota(e_1^*) &= \iota(e_1) \fatsemi e_1^* \cup \{ 1 \}
\end{align*}
\end{definition}

We can validate that the expressions in $\iota(e)$ together make up $e$, as follows.
\begin{lemma}%
\label{lemma:iota-correct}
Let $e \in \mathbb{E}$.
Now $e \equiv \sum_{e' \in \iota(e)} e'$.
\end{lemma}
\begin{proof}
The claim holds by definition in all three base cases.
The inductive arguments are as follows.
\begin{itemize}
    \item
    If $e = e_1 + e_2$, then we derive that
    \[
        e_1 + e_2
            \equiv \sum_{e_1' \in \iota(e_1)} e_1 + \sum_{e_2' \in \iota(e_2)} e_2'
            \equiv \sum_{e' \in \iota(e_1) \cup \iota(e_2)} e'
            = \sum_{e' \in \iota(e_1 + e_2)} e'
    \]
    \item
    If $e = e_1 \cdot e_2$, then we calculate
    \[
        e_1 \cdot e_2
            \equiv \Bigl( \sum_{e_1' \in \iota(e_1)} e_1' \Bigr) \cdot e_2
            \equiv \sum_{e_1' \in \iota(e_1)} e_1' \cdot e_2
            = \sum_{e' \in \iota(e_1 \cdot e_2)} e'
    \]
    \item
    If $e = e_1^*$, then we have
    \[
        e_1^*
            \equiv 1 + e_1 \cdot e_1^*
            \equiv 1 + \sum_{e_1' \in \iota(e_1)} e_1 \cdot e_1^*
            \equiv \sum_{e' \in \iota(1 + e_1 \cdot e_1^*)} e'
            = \sum_{e' \in \iota(e_1^*)} e'
    \]
\end{itemize}

\vspace{-8mm}
\end{proof}

\noindent
We could now try to package these parts into an automaton $(\mathbb{E}, \partial, \iota(e), \mathbb{F})$ for each expression $e$.
Because we work with finite automata, a little more work is necessary to identify the finite set of expressions that are relevant (i.e., represented by reachable states) for a starting expression $e$.

\begin{definition}
We define $\rho: \mathbb{E} \to \mathcal{P}(\mathbb{E})$ recursively, as follows:
\begin{align*}
    \rho(0) &= \emptyset
        & \rho(e_1 + e_2) &= \rho(e_1) \cup \rho(e_2) \\
    \rho(1) &= \{ 1 \}
        & \rho(e_1 \cdot e_2) &= \rho(e_1) \fatsemi e_2 \cup \rho(e_2) \\
    \rho(\ltr{a}) &= \{ \ltr{a}, 1 \}
        & \rho(e_1^*) &= \rho(e_1) \fatsemi e_1^* \cup \{ 1 \}
\end{align*}
\end{definition}

\begin{example}
If we take $e = \ltr{a} \cdot {(\ltr{b} \cdot \ltr{a})}^*$, then we can calculate that
\begin{align*}
    \rho(e)
        &= \rho(\ltr{a}) \fatsemi {(\ltr{b} \cdot \ltr{a})}^* \cup \rho({(\ltr{b} \cdot \ltr{a})}^*) \\
        &= \{ \ltr{a} \cdot {(\ltr{b} \cdot \ltr{a})}^*,\; 1 \cdot {(\ltr{b} \cdot \ltr{a})}^* \} \cup \rho({(\ltr{b} \cdot \ltr{a})}^*) \\
        &= \{ \ltr{a} \cdot {(\ltr{b} \cdot \ltr{a})}^*,\; 1 \cdot {(\ltr{b} \cdot \ltr{a})}^* \} \cup \{ 1 \} \cup \rho(\ltr{b} \cdot \ltr{a}) \fatsemi {(\ltr{b} \cdot \ltr{a})}^* \\
        &= \{ \ltr{a} \cdot {(\ltr{b} \cdot \ltr{a})}^*,\; 1 \cdot {(\ltr{b} \cdot \ltr{a})}^* \} \cup \{ 1 \} \cup (\rho(\ltr{b}) \fatsemi \ltr{a} \cup \rho(\ltr{a})) \fatsemi {(\ltr{b} \cdot \ltr{a})}^* \\
        &= \{ \ltr{a} \cdot {(\ltr{b} \cdot \ltr{a})}^*,\; 1 \cdot {(\ltr{b} \cdot \ltr{a})}^* \} \cup \{ 1 \} \cup \{ \ltr{b} \cdot \ltr{a},\; 1 \cdot \ltr{a},\; \ltr{a},\; 1 \} \fatsemi {(\ltr{b} \cdot \ltr{a})}^* \\
        &= \{ \ltr{a} \cdot {(\ltr{b} \cdot \ltr{a})}^*,\; 1 \cdot {(\ltr{b} \cdot \ltr{a})}^*,\; 1 ,\; \ltr{b} \cdot \ltr{a}\cdot {(\ltr{b} \cdot \ltr{a})}^*,\; 1 \cdot \ltr{a} \cdot {(\ltr{b} \cdot \ltr{a})}^* \}
\end{align*}
Note how there is some redundancy in the set above because the semantics of some expressions (e.g., $\ltr{a} \cdot {(\ltr{b} \cdot \ltr{a})}^*$ and $1 \cdot \ltr{a} \cdot {(\ltr{b} \cdot \ltr{a})}^*$) clearly coincide.
This is fine for us, because we just want to use Antimirov automata in proofs; if we wanted to optimize for the number of states, then $\rho$ needs some refinement.
\end{example}

\smallskip
With this function in hand, we can verify that it fits all of the requirements of the state space of an automaton with respect to the other parts identified above.

\begin{lemma}%
\label{lemma:automaton-well-defined}
For all $e \in \mathbb{E}$, the set $\rho(e)$ is finite and closed under $\partial$, i.e., if $e' \in \rho(e)$ and $e'' \in \partial(e', \ltr{a})$, then $e'' \in \rho(e)$ as well.
Furthermore, we have that $\iota(e) \subseteq \rho(e)$.
\end{lemma}
\begin{proof}
The proof of the first property proceeds by induction on $e$; we omit the arguments towards finiteness, as this holds for the base cases and is clearly preserved by both $\cup$ and $\fatsemi$.
For closure under $\partial$, we show more generally that if $e' \in \rho(e)$, then $\rho(e') \subseteq \rho(e)$, and that $\partial(e, \ltr{a}) \subseteq \rho(e)$ for all $\ltr{a} \in \Sigma$.
In the base, the claim holds by mere computation.
For the inductive step, there are three cases.
\begin{itemize}
    \item
    If $e = e_1 + e_2$, then $e' \in \rho(e_1)$ or $e' \in \rho(e_2)$; we assume the former without loss of generality.
    By induction, $\rho(e') \subseteq \rho(e_1)$, and so $\rho(e') \subseteq \rho(e)$ as well.
    For the second property, we note that $\partial(e, \ltr{a}) = \partial(e_1, \ltr{a}) \cup \partial(e_2, \ltr{a})$.
    By induction, the latter is a subset of $\rho(e_1) \cup \rho(e_2) = \rho(e)$.
    \item
    If $e = e_1 \cdot e_2$, then there are two cases.
    If $e' = e_1' \cdot e_2$ with $e_1' \in \rho(e_1)$, then $\rho(e') = \rho(e_1') \fatsemi e_2 \cup \rho(e_2)$.
    By induction, the latter is contained in $\rho(e_1) \fatsemi e_2 \cup \rho(e_2) = \rho(e)$.
    Otherwise, if $e' \in \rho(e_2)$, then by induction $\rho(e') \subseteq \rho(e_2) \subseteq \rho(e)$.
    For the second property, note $\partial(e, \ltr{a}) = \partial(e_1, \ltr{a}) \fatsemi e_2 \cup e_1 \star \partial(e_2, \ltr{a})$.
    By induction, this set is contained in $\rho(e_1) \fatsemi e_2 \cup \rho(e_2) = \rho(e)$.
    \item
    If $e = e_1^*$, we need only consider the case where $e' = e_1' \cdot e_1^*$ with $e_1' \in \rho(e_1)$.
    Here, $\rho(e') = \rho(e_1') \fatsemi e_1^* \cup \rho(e_1^*) \subseteq \rho(e_1) \fatsemi e_1^* \cup \rho(e_1^*) \subseteq \rho(e)$, with the second containment following by induction.
    For the second property, we derive that $\partial(e_1^*, \ltr{a}) = \partial(e_1, \ltr{a}) \fatsemi e_1^* \subseteq \rho(e_1) \fatsemi e_1^* \subseteq \rho(e)$.
\end{itemize}
The fact that $\iota(e) \subseteq \rho(e)$ follows by another straightforward induction on $e$.
\end{proof}

In light of the above, can treat $\partial$ as a function $\partial: \rho(e) \times \Sigma \to \mathcal{P}(\rho(e))$ with impunity.
If we now write $\mathbb{F}_e$ for $\mathbb{F} \cap \rho(e)$, then we can define Antimirov automata as follows.

\begin{definition}
Let $e \in \mathbb{E}$.
We write $A_e$ for the \emph{Antimirov automaton} $(\rho(e), \partial, \iota(e), \mathbb{F}_e)$.
\end{definition}

\begin{example}
If $e = \ltr{a} \cdot {(\ltr{b} \cdot \ltr{a})}^*$, then $A_e$ is the automaton depicted in \Cref{figure:antimirov-automaton}.
Note how there are two states not reachable from the initial state ($1$ and $\ltr{b} \cdot \ltr{a} \cdot {(\ltr{b} \cdot \ltr{a})}^*$).
As it turns out, these additional states do not matter for what we want to prove, so we keep them for the sake of simplicity.
\end{example}

\begin{figure}
    \centering
    \begin{tikzpicture}[every initial by arrow/.style={-latex}]
        \node[draw,rounded corners=1mm,initial,initial above,initial text={}] (abastar) {$\ltr{a} \cdot {(\ltr{b} \cdot \ltr{a})}^*$};
        \node[draw,accepting,rounded corners=1mm,right=of abastar] (onebastar) {$1 \cdot {(\ltr{b} \cdot \ltr{a})}^*$};
        \node[draw,rounded corners=1mm,right=of onebastar] (oneabastar) {$1 \cdot \ltr{a} \cdot {(\ltr{b} \cdot \ltr{a})}^*$};
        \node[draw,rounded corners=1mm,right=of oneabastar] (babastar) {$\ltr{b} \cdot \ltr{a} \cdot {(\ltr{b} \cdot \ltr{a})}^*$};
        \node[draw,accepting,rounded corners=1mm,left=of abastar] (one) {$1$};
        \draw (abastar) edge[-latex] node[above] {$\ltr{a}$} (onebastar);
        \draw (onebastar.north east) edge[-latex,bend left] node[above] {$\ltr{b}$} (oneabastar.north west);
        \draw (oneabastar.south west) edge[-latex,bend left] node[below] {$\ltr{a}$} (onebastar.south east);
        \draw (babastar) edge[-latex] node[above] {$\ltr{b}$} (oneabastar);
    \end{tikzpicture}
    \caption{The Antimirov automaton $A_e$ where $e = \ltr{a} \cdot {(\ltr{b} \cdot \ltr{a})}^*$.}%
    \label{figure:antimirov-automaton}
\end{figure}
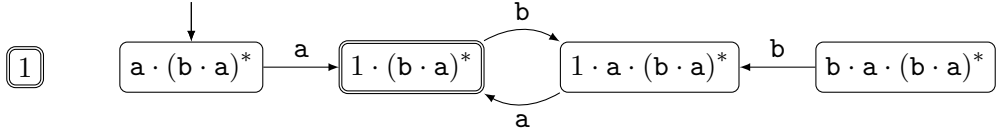

We end this section with two technical lemmas about Antimirov automata.
The first says that for all $e \in \mathbb{E}$, the language of $A_e$ contains all of the words in $\widehat{\ell}(e)$.
This is essentially one half of the correctness theorem for Antimirov's construction~\cite{antimirov-1996}; the other half is also not too hard to show, but we argue this weaker version because it is sufficient for our purposes.

\begin{lemma}[{One half of~\cite[Theorem~4.1]{antimirov-1996}}]%
\label{lemma:antimirov-correct}
Let $e \in \mathbb{E}$.
Now $\widehat{\ell}(e) \subseteq L(A_e)$.
\end{lemma}
\begin{proof}
We start by establishing two useful properties: first, if $\epsilon \in \widehat{\ell}(e)$, then $e \in \mathbb{F}$; second, if $\ltr{a}w \in \widehat{\ell}(e)$, then $w \in \widehat{\ell}(e')$ for some $e' \in \partial(e, \ltr{a})$.
The first claim follows by a straightforward induction on $e$.
As for the second claim, we also proceed by induction on $e$.
In the base, we need only consider the case where $e = \ltr{b}$ for some $\ltr{b} \in \Sigma$.
By the premise, we then know that $w = \epsilon$ and $\ltr{a} = \ltr{b}$, so we can choose $e' = 1$ to satisfy the claim.
For the inductive step, there are three cases.
\begin{itemize}
    \item
    If $e = e_1 + e_2$, then without loss of generality $\ltr{a}w \in \widehat{\ell}(e_1)$.
    By induction, we then have that $w \in \widehat{\ell}(e_1')$ for some $e_1' \in \partial(e_1, \ltr{a}) \subseteq \partial(e, \ltr{a})$, completing the claim.
    \item
    If $e = e_1 \cdot e_2$, then there are two subcases, depending on where $\ltr{a}$ comes from.
    If $\epsilon \in \widehat{\ell}(e_1)$ and $\ltr{a}w \in \widehat{\ell}(e_2)$, then by induction $w \in \widehat{\ell}(e_2')$ for some $e_2' \in \partial(e_2, \ltr{a})$.
    Furthermore, by the first property, we have that $e_1 \in \mathbb{F}$.
    We then know that $e_2' \in \partial(e, \ltr{a})$.
    Otherwise, if $w = w_1w_2$ such that $\ltr{a}w_1 \in \widehat{\ell}(e_1)$ and $w_2 \in \widehat{\ell}(e_2)$, then by induction $w_1 \in \widehat{\ell}(e_1')$ for some $e_1' \in \partial(e_1, \ltr{a})$.
    We can then choose $e' = e_1' \cdot e_2 \in \partial(e, \ltr{a})$ to satisfy the claim.
    \item
    If $e = e_1^*$, then $\ltr{a}w = w_1 \cdots w_n$ such that $w_1, \dots, w_n \in \widehat{\ell}(e_1)$.
    Without loss of generality, each of these $w_i$ is non-empty, and so $w_1 = \ltr{a}w_1'$ for some word $w_1'$; furthermore, $w_2 \cdots w_n \in \widehat{\ell}(e_1^*)$.
    Now we have by induction that $w_1' \in \widehat{\ell}(e_1')$ for some $e_1' \in \partial(e_1, \ltr{a})$.
    We can then choose $e' = e_1' \cdot e_1^*$ to find that $w = w_1' w_2 \cdots w_n \in \widehat{\ell}(e')$ and $e' \in \partial(e_1^*, \ltr{a})$.
\end{itemize}

With these properties in hand, we can now prove that for all $e' \in \rho(e)$, if $w \in \widehat{\ell}(e')$, then $w \in L(A_e, e')$, by induction on $w$.
In the base, where $w = \epsilon$, the first property tells us that $e' \in \mathbb{F}$, and so $\epsilon \in L(A_e, e')$ by definition.
Otherwise, we can write $w = \ltr{a}w'$, which by the second property tells us that $w' \in \widehat{\ell}(e'')$ for some $e'' \in \partial(e', \ltr{a})$, whence $w' \in L(A_e, e'')$ by induction --- but then $w = \ltr{a}w' \in L(A_e, e')$ by definition of $L(A_e, -)$.
Now, if $w \in \widehat{\ell}(e)$, then $w \in \widehat{\ell}(e')$ for some $e' \in \iota(e)$ by \Cref{lemma:iota-correct}, and hence $w \in L(A_e, e')$ by the above, which means that $w \in L(A_e)$ as claimed.
\end{proof}

The final property of Antimirov automata that we need is that the solution to $A_e$ is below $e$.
As before, it is possible to show the converse, i.e., $e \leqq \soli{A_e}$~\cite{kappe-2023}.
However, this requires some additional formal tools, and is in fact not necessary to show the finite model property.
The property is an easy consequence of \Cref{lemma:half-fundamental}, and in light of similar results for other constructions~\cite{kozen-2001,jacobs-2006}, not surprising.

\begin{lemma}%
\label{lemma:antimirov-solution-upper}
For all $e \in \mathbb{E}$, it holds that $\soli{A_e} \leqq e$.
\end{lemma}
\begin{proof}
Let $s: \rho(e) \to \mathbb{E}$ be the injection.
\Cref{lemma:half-fundamental} shows that $s$ is a solution to $A_e$, meaning that for $e' \in \rho(e)$ we have $\sol{A_e}(e') \leqq s(e') = e'$, by \Cref{theorem:computing-solutions}.
Using \Cref{lemma:iota-correct}, we then conclude
\[
    \soli{A_e}
        \equiv \sum_{e' \in \iota(e)} \sol{A_e}(e')
        \leqq \sum_{e' \in \iota(e)} e'
        \equiv e
\]

\vspace{-8mm}
\end{proof}

\section{The Finite Model Property}%
\label{section:prove-fmp}

Recall that our objective was to prove the finite model property.
To this end, we shall derive a finite Kleene algebra for each expression, whose properties can then be used to conclude the finite model property, thus arriving at our second contribution.
We already saw how an expression gives rise to an automaton and its transformation automata.
As it happens, the states of transformation automata have the internal structure of a \emph{monoid} --- indeed, this was the original motivation for the construction~\cite{mcnaughton-papert-1968}.

\begin{definition}
A \emph{monoid} is a tuple $(M, \cdot, 1)$ where $M$ is a set, $\cdot$ is a binary operator and $1 \in M$ such that the following hold for all $m_1, m_2, m_3 \in M$:
\begin{mathpar}
    m_1 \cdot (m_2 \cdot m_3) = (m_1 \cdot m_2) \cdot m_3
    \and
    m_1 \cdot 1 = m_1
    \and
    1 \cdot m_1 = m_1
\end{mathpar}
A function $h: \Sigma \to M$ gives rise to the function $\widetilde{h}: \Sigma^* \to M$, defined by\[
    \widetilde{h}(\ltr{a}_1 \cdots \ltr{a}_n) = h(\ltr{a_1}) \cdot \cdots \cdot h(\ltr{a_n})
\]
As for KAs, we may denote a general monoid $(M, \cdot, 1)$ with its carrier $M$.
\end{definition}

If $A = (Q, \delta, I, F)$ is an automaton, then $\mathcal{R}(Q)$, the state space of its transformation automata, has a monoidal structure: the operator is relational composition, and the unit is the identity relation.

A monoid can be lifted to sets of its elements, which lets us define the other operators of a KA\@.

\begin{lemma}[see~\cite{palka-2005}]%
\label{lemma:monoid-to-ka}
If $(M, \cdot, 1)$ is a monoid, then $(\mathcal{P}(M), \cup, \otimes, {}^{\circledast}, \emptyset, \{ 1 \})$ is a star-continuous KA, where for $U, V \subseteq M$:
\begin{mathpar}
    U \otimes V = \{ m \cdot n : m \in U, n \in V \}
    \and
    U^{\circledast} = \{ u_1 \cdots u_n : u_1, \dots, u_n \in U \}
\end{mathpar}
\end{lemma}

As an example, note that $\Sigma^*$ has a monoidal structure (with concatenation as operation and $\epsilon$ as unit), and that $\mathcal{P}(\Sigma^*)$ as derived from the above is precisely the KA of languages.

\begin{remark}
The algebraically minded reader may interject that the monoid of relations has an even richer structure, namely that of a \emph{idempotent semiring}, with the additive operator given by union.
A finite semiring can directly be treated a (star-continuous) KA, without involving the powerset: simply definite $x^*$ as the (finite!) sum of $x^n$ for all $n \in \mathbb{N}$.
The problem with this approach is that the resulting KA may fail to satisfy \Cref{lemma:interp-upper,lemma:interp-lower}, which we need to prove the finite model property.
\end{remark}

Now, given an expression $e \in \mathbb{E}$, a monoid $(M, \cdot, 1)$, and a map $h: \Sigma \to M$, we have two ways of interpreting $e$ inside $\mathcal{P}(M)$.
On the one hand, we can lift $\ltr{a} \mapsto \{ h(\ltr{a}) \}$ to obtain a map $\mathbb{E} \to \mathcal{P}(M)$.
On the other hand, we can map each $w \in \widehat{\ell}(e)$ to an element of $M$ via $\widetilde{h}: \Sigma^* \to M$ to obtain a subset of $M$.
The next lemma shows that these interpretations are actually the same; it can be thought of as a generalization of~\cite[Lemma~3.1]{palka-2005}, which covers the special case for the syntactic monoid.

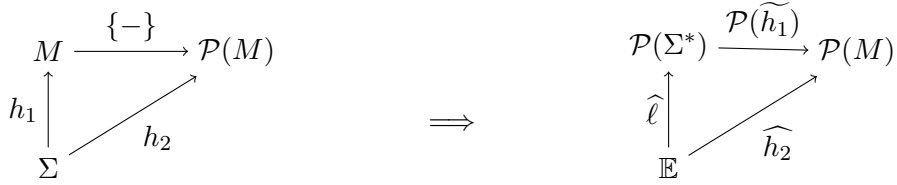
\begin{figure}
    \begin{mathpar}
        \begin{tikzpicture}[baseline=5mm]
            \node (Sigma) {$\Sigma$};
            \node[above=of Sigma] (M) {$M$};
            \node[right=15mm of M] (PM) {$\mathcal{P}(M)$};
            \draw (Sigma) edge[->] node[left] {$h_1$} (M);
            \draw (M) edge[->] node[above] {$\{ - \}$} (PM);
            \draw (Sigma) edge[->] node[below right] {$h_2$} (PM);
        \end{tikzpicture}
        \and
        \implies
        \and
        \begin{tikzpicture}[baseline=5mm]
            \node (Exp) {$\mathbb{E}$};
            \node[above=of Exp] (PSigmaStar) {$\mathcal{P}(\Sigma^*)$};
            \node[right=15mm of M] (PM) {$\mathcal{P}(M)$};
            \draw (Exp) edge[->] node[left] {$\widehat{\ell}$} (PSigmaStar);
            \draw (PSigmaStar) edge[->] node[above] {$\mathcal{P}(\widetilde{h_1})$} (PM);
            \draw (Exp) edge[->] node[below right] {$\widehat{h_2}$} (PM);
        \end{tikzpicture}
    \end{mathpar}
    \caption{
        Diagrammatic representation of \Cref{lemma:monoid-to-ka-interp}.
        In the diagram on the left, $\{ - \}$ is the map that sends an element of $m$ to the singleton set $\{ m \}$; on the right, $\mathcal{P}(\widetilde{h_1})$ is the pointwise application of $\widehat{h_1}: \Sigma^* \to M$.
    }%
    \label{figure:triangle-diagrams}
\end{figure}

\begin{lemma}%
\label{lemma:monoid-to-ka-interp}
Let $M$ be a monoid and let $\mathcal{P}(M)$ be the KA obtained from it, per \Cref{lemma:monoid-to-ka}.
Also, let $h_1: \Sigma \to M$ and $h_2: \Sigma \to \mathcal{P}(M)$ be such that if $\ltr{a} \in \Sigma$, then $h_2(\ltr{a}) = \{ h_1(\ltr{a}) \}$.
Then for all $e \in \mathbb{E}$:
\[
    \widehat{h_2}(e) = \{ \widetilde{h_1}(w) : w \in \widehat{\ell}(e) \}
\]
In categorical terms, if the diagram on the left in \Cref{figure:triangle-diagrams} commutes, then so does the one on the right.
\end{lemma}
\begin{proof}
For the inclusion from right to left, note that if $w \in \widehat{\ell}(e)$, then $\widetilde{h_1}(w) \in \widehat{h_2}(w) \subseteq \widehat{h_2}(e)$ by \Cref{lemma:embed-sem}.
For the other inclusion, note that for all $w \in \widehat{\ell}(e)$ we have $\widehat{h_2}(w) = \{ \widetilde{h_1}(w) \} \subseteq \{ \widetilde{h_1}(w) : w \in \widehat{\ell}(e) \}$, and so $\widehat{h_2}(e) \subseteq \{ \widetilde{h_1}(w) : w \in \widehat{\ell}(e) \}$ by \Cref{lemma:cont}.
\end{proof}

In particular, every expression $e$ gives rise to an automaton $A_e$, which in turn gives a monoid $(\mathcal{R}(\rho(e)), \circ, \id)$ that yields a KA with carrier $\mathcal{P}(\mathcal{R}(\rho(e)))$;\@ we write $\mathcal{K}_e$ for this KA, and $h_e$ for the map from $\Sigma$ to $\mathcal{K}_e$ given by $h_e(\ltr{a}) = \{ \partial_\ltr{a} \}$.
This is precisely the KA that we will need to prove the FMP\@.

\begin{example}
If we set $e = \ltr{a} \cdot {(\ltr{b} \cdot \ltr{a})}^*$, then $A_e$ is as in \Cref{figure:antimirov-automaton}.
The accompanying transformation automaton happens to be isomorphic to the one drawn in \Cref{figure:transformation-automaton-alt}.
The state space of this transformation automaton induces a monoid $(\mathcal{R}(\rho(e)), \circ, \id)$, and hence a KA $\mathcal{K}_e = (\mathcal{P}(\mathcal{R}(\rho(e))), \cup, \otimes, {}^\oast, \emptyset, \{\id\})$.

We can interpret the expression $\ltr{a} \cdot \ltr{b}^* \cdot \ltr{a}$ in $\mathcal{K}_e$; \Cref{lemma:monoid-to-ka-interp} then tells us that
\[
    \widehat{h_e}(\ltr{a} \cdot \ltr{b}^* \cdot \ltr{a})
        = \{ \widetilde{h_e}(w) : w \in \widehat{\ell}(\ltr{a} \cdot \ltr{b}^* \cdot \ltr{a}) \}
        = \{ \partial_\ltr{aa}, \partial_\ltr{aba}, \partial_\ltr{abba}, \dots \}
        = \{ \emptyset, \partial_\ltr{a} \}
\]
We can also interpret $e$ itself inside $\mathcal{K}_e$ to find that
\[
    \widehat{h_e}(e)
        = \{ \widetilde{h_e}(w) : w \in \widehat{\ell}(e) \}
        = \{ \partial_\ltr{a}, \partial_\ltr{aba}, \partial_\ltr{ababa}, \dots \}
        = \{ \partial_\ltr{a} \}
\]
\end{example}

The KA $\mathcal{K}_e$ thus gives us a way to obtain a bunch of relations on the state space of $A_e$, the Antimirov automaton of $e$.
Now, each of these relations $R$ represents a set of words that have the same effect on $A_e$, and we know that we can obtain an expression for all of those words by solving the transformation automaton with $R$ as an accepting state.
But then, in light of \Cref{lemma:monoid-to-ka-interp}, we should expect that all of the relations in $\widehat{h_e}(e)$ are of the form $\partial_w$, for some $w \in \widehat{\ell}(e)$; these must then correspond to expressions whose words bring about a transformation from an initial state of $A_e$ to some accepting state.
This line of reasoning can again be reflected on the level of equational reasoning, as follows.

\begin{lemma}%
\label{lemma:interp-upper}
Let $e \in \mathbb{E}$.
If $R \in \widehat{h_e}(e)$, then $\soli{A_e[R]} \leqq e$.
\end{lemma}
\begin{proof}
By \Cref{lemma:monoid-to-ka-interp}, we have that $R = \partial_w$ for some $w \in \widehat{\ell}(e)$.
But then $e' \mathrel{\partial_w} e''$ for some $e' \in \iota(e)$ and $e'' \in \mathbb{F}_e$, by \Cref{lemma:antimirov-correct}.
By \Cref{lemma:solve-transformation-aut-approximate-below}, it follows that $\soli{A_e[R]} \leqq \sol{A_e}(e')$, and since $e' \in \iota(e)$ also $\sol{A_e}(e') \leqq \soli{A_e}$.
Finally, using \Cref{lemma:antimirov-solution-upper}, we conclude by noting that $\soli{A_e} \leqq e$.
\end{proof}

\clearpage
We can also think of the expressions obtained by solving the automaton $A_e[R]$ for $R \in \widehat{h_e}(f)$ as partitioning the set of all words: each $w \in \Sigma^*$ corresponds to precisely one transformation on the state space of $A_e$.
If we turn to \Cref{lemma:monoid-to-ka-interp} once more, each of the relations in $\widehat{h_e}(f)$ is of the form $\partial_w$ for $w \in \widehat{\ell}(f)$; because we can read $w$ when moving from $\id$ to $\partial_w$ in $A_e[\partial_w]$, we have that $\soli{A_e[\partial_w]}$ must also include $w$.
In total, all of the $\soli{A_e[\partial_w]}$ together must include all of the words denoted by $f$ (and possibly more), in essence approximating $f$ from above using the expressions we obtained from $e$.
The algebraic counterpart to this line of reasoning goes as follows.

\begin{lemma}%
\label{lemma:interp-lower}
Let $e, f \in \mathbb{E}$.
The following holds:
\[
    f \leqq \sum_{R \in \widehat{h_e}(f)} \soli{A_e[R]}
\]
\end{lemma}
\begin{proof}
We proceed by induction on $f$.
In the base, there are three cases.
If $f = 0$, then the claim holds immediately.
If $f = 1$, then note that $1 \leqq \sol{A_e[\id]}(\id) \leq \soli{A_e[\id]}$; the claim then follows because $\widehat{h_e}(f) = \{ \id \}$ by definition.
If $f = \ltr{a}$ for some $\ltr{a} \in \Sigma$, then $f \leqq \soli{A_e[\partial_\ltr{a}]}$ by \Cref{lemma:solve-transformation-aut-letter}.
The claim then follows because $\widehat{h_e}(f) = \{ \partial_\ltr{a} \}$.
There are three inductive cases.
\begin{itemize}
    \item
    If $f = f_1 + f_2$, then we calculate as follows:
    \begin{align*}
        f_1 + f_2
            &\leqq \sum_{R \in \widehat{h_e}(f_1)} \soli{A_e[R]} + \sum_{R \in \widehat{h}(f_2)} \soli{A_e[R]} \tag{IH} \\
            &\equiv \sum_{R \in \widehat{h_e}(f_1) \cup \widehat{h}(f_2)} \soli{A_e[R]} \tag{merge sums} \\
            &= \sum_{R \in \widehat{h_e}(f_1 + f_2)} \soli{A_e[R]} \tag{def.\ $\widehat{h_e}$}
    \intertext{
    \item
    If $f = f_1 \cdot f_2$, then we can derive
    }
        f_1 \cdot f_2
            &\leqq \Bigl( \sum_{R \in \widehat{h_e}(f_1)} \soli{A_e[R]} \Bigr) \cdot \Bigl( \sum_{R \in \widehat{h_e}(f_2)} \soli{A_e[R]} \Bigr) \tag{IH} \\
            &\leqq \sum_{\substack{R_1 \in \widehat{h_e}(f_1) \\ R_2 \in \widehat{h_e}(f_2)}} \soli{A_e[R_1]} \cdot \soli{A[R_2]} \tag{distributivity} \\
            &\leqq \sum_{R \in \widehat{h_e}(f_1 \cdot f_2)} \soli{A_e[R]} \tag{\Cref{lemma:solve-transformation-aut-compose}}
    \intertext{
    \item
    If $f = f_1^*$, then we find that:
    }
        1 + f_1 \cdot \Bigl( \sum_{R \in \widehat{h_e}(f_1^*)} \soli{A_e[R]} \Bigr)
            &\leqq 1 + \Bigl( \sum_{R_1 \in \widehat{h_e}(f_1)} \soli{A_e[R_1]} \Bigr) \cdot \Bigl( \sum_{R \in \widehat{h_e}(f_1^*)} \soli{A_e[R]} \Bigr) \tag{IH} \\
            &\leqq 1 + \sum_{\substack{R_1 \in \widehat{h_e}(f_1) \\ R \in \widehat{h_e}(f_1^*)}} \soli{A_e[R_1]} \cdot \soli{A_e[R]} \tag{distributivity} \\
            &\leqq 1 + \sum_{R \in \widehat{h_e}(f_1 \cdot f_1^*)} \soli{A_e[R]} \tag{\Cref{lemma:solve-transformation-aut-compose}} \\
            &\equiv \sum_{R \in \widehat{h_e}(1 + f_1 \cdot f_1^*)} \soli{A_e[R]} \tag{def.\ $\soli{-}$, merge sums} \\
            &= \sum_{R \in \widehat{h_e}(f_1^*)} \soli{A_e[R]} \tag{$f_1^* \equiv 1 + f_1 \cdot f_1^*$} 
    \end{align*}
    By the left fixpoint rule, it then follows that $f \equiv f_1^* \cdot 1 \leqq \sum_{R \in \widehat{h_e}(f_1^*)} \soli{A_e[R]}$.
\end{itemize}

\vspace{-8mm}
\end{proof}

\noindent
We then have all of the machinery necessary to conclude our proof of the finite model property.

\restatefmp*
\begin{proof}
By the premise, $\mathcal{K}_e \models e = f$, i.e., $\widehat{h_e}(e) = \widehat{h_e}(f)$, and so by \Cref{lemma:interp-lower,lemma:interp-upper}:
\[
    f
        \leqq \sum_{R \in \widehat{h_e}(f)} \soli{A_e[R]}
        = \sum_{R \in \widehat{h_e}(e)} \soli{A_e[R]}
        \leqq e
\]
By a similar argument, $e \leqq f$, and thus $e \equiv f$.
\end{proof}

\section{Discussion}%
\label{section:discussion}

We leave the reader with some final thoughts regarding our formalization and possible further work.

\paragraph{Coq formalization}
We have formalized~\cite{proofs} all of our results in Coq~\cite{bertot-casteran-2004,coq-web}; a mapping of the formal statements in this article to the Coq development is given in \Cref{table:mapping-to-development}.
The trusted base comes down to (1)~the axioms of the Calculus of Inductive Constructions, (2)~injectivity of dependent equality (equivalent to Streicher's axiom K~\cite{streicher-1993}), (3)~dependent functional extensionality, and (4)~propositional extensionality.
The latter two are a result of our encoding of subsets and relational Kleene algebras respectively, and can most likely be factored out with better data structures.

\begin{table}
    \centering
    \begin{tabular}{lll}
    \textbf{Property}                                   & \textbf{Coq symbol}                                       & \textbf{Filename} \\\toprule
    \Cref{theorem:finite-model-property}                & \texttt{finite\_model\_property}                          & \multirow{3}{*}{\texttt{Main.v}} \\
    \Cref{theorem:relational-model-property}            & \texttt{relational\_model\_property}                      & \\
    \Cref{theorem:language-completeness}                & \texttt{completeness}                                     & \\
    \midrule
    \Cref{lemma:cont}                                   & \texttt{kleene\_star\_continuous\_interp\_lower}          & \multirow{4}{*}{\texttt{ModelTheory.v}} \\
    \Cref{lemma:embed-sem}                              & \texttt{kleene\_star\_continuous\_interp\_upper}          & \\
    \Cref{lemma:lang-vs-cont}                           & \texttt{preserve\_language\_to\_star\_continuous}         & \\
    \Cref{lemma:fin-rel-vs-lang}                        & \texttt{preserve\_finite\_relational\_to\_language}       & \\
    \midrule
    \Cref{theorem:computing-matrix-solutions}           & \texttt{compute\_solution\_least\_solution}               & \multirow{2}{*}{\texttt{Solve.v}} \\
    \Cref{theorem:computing-solutions}                  & \texttt{compute\_automaton\_solution\_least\_solution}    & \\
    \midrule
    \Cref{lemma:solve-transformation-aut-letter}        & \texttt{automaton\_relation\_solution\_letter}            & \multirow{3}{*}{\texttt{Transformation.v}} \\
    \Cref{lemma:solve-transformation-aut-shift}         & \texttt{automaton\_relation\_solution\_shift}             & \\
    \Cref{lemma:solve-transformation-aut-compose}       & \texttt{automaton\_relation\_solution\_merge}             & \\
    \midrule
    \Cref{lemma:half-fundamental}                       & \texttt{nullable\_one, derive\_step}                      & \multirow{5}{*}{\texttt{Antimirov.v}} \\
    \Cref{lemma:iota-correct}                           & \texttt{initial\_reconstruct}                             & \\
    \Cref{lemma:automaton-well-defined}                 & \texttt{derivative\_finite} (\textdagger)                 & \\
    \Cref{lemma:antimirov-correct}                      & \texttt{automaton\_antimirov\_accepts}                    & \\
    \Cref{lemma:antimirov-solution-upper}               & \texttt{antimirov\_solution\_upper\_bound}                & \\
    \midrule
    \Cref{lemma:monoid-to-ka}                           &  \texttt{monoid\_to\_kleene\_algebra}                     & \multirow{5}{*}{\texttt{CanonicalModel.v}} \\
    \multirow{2}{*}{\Cref{lemma:monoid-to-ka-interp}}   &  \texttt{kleene\_interp\_witness\_apply}                  & \\
                                                        &  \texttt{kleene\_interp\_witness\_construct}              & \\
    \Cref{lemma:interp-upper}                           &  \texttt{automaton\_kleene\_algebra\_interp\_upper}       & \\
    \Cref{lemma:interp-lower}                           &  \texttt{automaton\_kleene\_algebra\_interp\_lower}       &
    \end{tabular}
    \caption{Mapping of statements in this article to the Coq development. In the case of \Cref{lemma:automaton-well-defined} (\textdagger), only the claim about finiteness is proven in explicitly; the other properties are true by virtue of the encoding of derivatives.}%
    \label{table:mapping-to-development}
\end{table}

\paragraph{Representation theory}
The lemmas in \Cref{section:model-theory} suggest another question: is it possible to represent every (finite or star-continuous) KA using relations?
In other words, is it the case that every such KA is a subalgebra of a relational KA\@?
Such a result would amount to a \emph{representation theorem} for KAs, along the same lines as Cayley's representation theorem for groups.
Representation theory is a well established branch of algebra, but little seems to be known about it in the setting of KA beyond an investigation by Kozen~\cite{kozen-2006}, who showed that KAs \emph{with tests} that satisfy certain additional properties can be represented using relations.
We hope to investigate representations of KAs in future work.

\paragraph{Concurrent Kleene Algebra}
\emph{Concurrent Kleene Algebra} (CKA) was proposed to expand the reasoning power of KA to include programs with concurrency~\cite{hoare-moeller-struth-wehrman-2009}.
Several variants of CKA enjoy completeness properties, including bi-Kleene algebra (BKA)~\cite{laurence-struth-2014}, which includes parallel composition and its corresponding analogue to the Kleene star, and weak CKA~\cite{laurence-struth-2017-arxiv,kappe-brunet-silva-zanasi-2018}, which does not include the parallel Kleene star but does include the \emph{exchange law}, which allows programs to be interleaved.

Completeness for the full calculus of CKA, including both the parallel star operator and the exchange law, remains an open question, and it is known that existing techniques do not generalize to that setting~\cite{kappe-brunet-silva-zanasi-2018}.
Perhaps the approach in this article could provide a different route that does generalize to CKA;\@ this would require giving an algebraic angle an operational model for CKA~\cite{kappe-brunet-luttik-silva-zanasi-2019,jipsen-moshier-2016,fahrenberg-johansen-struth-ziemanski-2022}.

\paragraph{Guarded Kleene Algebra with Tests}
The fragment of KAT where non-deterministic choice and loops are guarded by tests, called Guarded Kleene Algebra with Tests (GKAT)~\cite{smolka-foster-hsu-kappe-kozen-silva-2020,schmid-kappe-kozen-silva-2021}, has favorable complexity properties.
Moreover, GKAT admits a set of axioms that are complete w.r.t.\ its language (resp.\ relational, probabilistic) model, but this set is infinite as a result of an axiom scheme.

We wonder whether the techniques discussed here could be applied to arrive at a more satisfactory completeness result.
To start answering this question, one would first have to devise an analogue to transformation automata and monoids for GKAT\@.

\subsubsection*{Acknowledgements}
The author wishes to thank Nick Bezhanishvili for his suggestion to investigate the FMP for KA, Alexandra Silva for general advice about this manuscript, Hans Lei\ss{} for the discussions about the finite relational model property, and the anonymous reviewers for their careful comments.

\bibliographystyle{fundam}
\bibliography{main.bib}

\begin{thebibliography}{10}
\providecommand{\url}[1]{\texttt{#1}}
\providecommand{\urlprefix}{URL }
\expandafter\ifx\csname urlstyle\endcsname\relax
  \providecommand{\doi}[1]{doi:\discretionary{}{}{}#1}\else
  \providecommand{\doi}{doi:\discretionary{}{}{}\begingroup \urlstyle{rm}\Url}\fi
\providecommand{\eprint}[2][]{\url{#2}}

\bibitem{anderson-foster-guha-etal-2014}
Anderson CJ, Foster N, Guha A, Jeannin J, Kozen D, Schlesinger C, Walker D.
\newblock {NetKAT}: semantic foundations for networks.
\newblock In: POPL. 2014 pp. 113--126.
\newblock \doi{10.1145/2535838.2535862}.

\bibitem{antimirov-1996}
Antimirov VM.
\newblock Partial Derivatives of Regular Expressions and Finite Automaton Constructions.
\newblock \emph{Theor. Comput. Sci.}, 1996.
\newblock \textbf{155}(2):291--319.
\newblock \doi{10.1016/0304-3975(95)00182-4}.

\bibitem{backhouse-1975}
Backhouse R.
\newblock Closure algorithms and the star-height problem of regular languages.
\newblock Ph.D. thesis, University of London, 1975.
\newblock \urlprefix\url{http://hdl.handle.net/10044/1/22243}.

\bibitem{bertot-casteran-2004}
Bertot Y, Cast{\'{e}}ran P.
\newblock Interactive Theorem Proving and Program Development - Coq'Art: The Calculus of Inductive Constructions.
\newblock Texts in Theoretical Computer Science. An {EATCS} Series. 2004.
\newblock ISBN 978-3-642-05880-6.
\newblock \doi{10.1007/978-3-662-07964-5}.

\bibitem{boffa-1990}
Boffa M.
\newblock Une remarque sur les syst{\`{e}}mes complets d'identit{\'{e}}s rationnelles.
\newblock \emph{{RAIRO} Theor. Informatics Appl.}, 1990.
\newblock \textbf{24}:419--423.
\newblock \doi{10.1051/ita/1990240404191}.

\bibitem{bonchi-pous-2013}
Bonchi F, Pous D.
\newblock Checking NFA equivalence with bisimulations up to congruence.
\newblock In: POPL. 2013 pp. 457--468.
\newblock \doi{10.1145/2429069.2429124}.

\bibitem{braibant-pous-2012}
Braibant T, Pous D.
\newblock Deciding Kleene Algebras in Coq.
\newblock \emph{Log. Methods Comput. Sci.}, 2012.
\newblock \textbf{8}(1).
\newblock \doi{10.2168/LMCS-8(1:16)2012}.

\bibitem{brzozowski-1964}
Brzozowski JA.
\newblock Derivatives of Regular Expressions.
\newblock \emph{J. {ACM}}, 1964.
\newblock \textbf{11}(4):481--494.
\newblock \doi{10.1145/321239.321249}.

\bibitem{burris-sankappanavar-1981}
Burris S, Sankappanavar {\relax HP}.
\newblock A Course in Universal Algebra.
\newblock Graduate Texts in Mathematics. Springer, 1981.
\newblock ISBN 978-0387905785.

\bibitem{cohen-kozen-smith-1996}
Cohen E, Kozen D, Smith F.
\newblock The Complexity of {K}leene Algebra with Tests.
\newblock Technical Report TR96-1598, 1996.
\newblock \urlprefix\url{https://hdl.handle.net/1813/7253}.

\bibitem{conway-1971}
Conway JH.
\newblock Regular Algebra and Finite Machines.
\newblock Chapman and Hall, Ltd., London, 1971.

\bibitem{coq-web}
{\relax {Coq} {Development} {Team}}.
\newblock The {Coq} Reference Manual, version 8.15, 2022.
\newblock Available electronically at \url{http://coq.inria.fr/doc}.

\bibitem{das-doumane-pous-2018}
Das A, Doumane A, Pous D.
\newblock Left-Handed Completeness for {K}leene algebra, via Cyclic Proofs.
\newblock In: LPAR. 2018 pp. 271--289.
\newblock \doi{10.29007/hzq3}.

\bibitem{fahrenberg-johansen-struth-ziemanski-2022}
Fahrenberg U, Johansen C, Struth G, Ziemianski K.
\newblock A Kleene Theorem for Higher-Dimensional Automata.
\newblock In: CONCUR. 2022 pp. 29:1--29:18.
\newblock \doi{10.4230/LIPIcs.CONCUR.2022.29}.

\bibitem{foster-struth-2015}
Foster S, Struth G.
\newblock On the Fine-Structure of Regular Algebra.
\newblock \emph{J. Autom. Reason.}, 2015.
\newblock \textbf{54}(2):165--197.
\newblock \doi{10.1007/s10817-014-9318-9}.

\bibitem{hoare-moeller-struth-wehrman-2009}
Hoare T, M{\"{o}}ller B, Struth G, Wehrman I.
\newblock Concurrent {K}leene Algebra.
\newblock In: CONCUR. 2009 pp. 399--414.
\newblock \doi{10.1007/978-3-642-04081-8_27}.

\bibitem{jacobs-2006}
Jacobs B.
\newblock A Bialgebraic Review of Deterministic Automata, Regular Expressions and Languages.
\newblock In: Algebra, Meaning, and Computation, Essays Dedicated to Joseph A. Goguen on the Occasion of His 65th Birthday. 2006 pp. 375--404.
\newblock \doi{10.1007/11780274_20}.

\bibitem{jipsen-moshier-2016}
Jipsen P, Moshier MA.
\newblock Concurrent {K}leene Algebra with tests and branching automata.
\newblock \emph{J. Log. Algebr. Meth. Program.}, 2016.
\newblock \textbf{85}(4):637--652.
\newblock \doi{10.1016/j.jlamp.2015.12.005}.

\bibitem{kappe-2023}
Kapp\'{e} T.
\newblock Completeness and the Finite Model Property for Kleene Algebra, Reconsidered.
\newblock In: RAMiCS. 2023 pp. 158--175.
\newblock \doi{10.1007/978-3-031-28083-2_10}.

\bibitem{proofs}
Kapp\'{e} T.
\newblock An Elementary Proof of the {FMP} for {K}leene Algebra ({C}oq Formalization), 2024.
\newblock \doi{10.5281/zenodo.13293233}.

\bibitem{kappe-brunet-silva-zanasi-2018}
Kapp\'{e} T, Brunet P, Silva A, Zanasi F.
\newblock Concurrent Kleene Algebra: Free Model and Completeness.
\newblock In: ESOP. 2018 pp. 856--882.
\newblock \doi{10.1007/978-3-319-89884-1_30}.

\bibitem{kappe-brunet-luttik-silva-zanasi-2019}
Kapp\'{e} T, Brunet P, Luttik B, Silva A, Zanasi F.
\newblock On series-parallel pomset languages: Rationality, context-freeness and automata.
\newblock \emph{J. Log. Algebr. Meth. Program.}, 2019.
\newblock \textbf{103}:130--153.
\newblock \doi{10.1016/j.jlamp.2018.12.001}.

\bibitem{kleene-1956}
Kleene SC.
\newblock Representation of Events in Nerve Nets and Finite Automata.
\newblock \emph{Automata Studies}, 1956.
\newblock pp. 3--41.

\bibitem{kozen-1990}
Kozen D.
\newblock On Kleene Algebras and Closed Semirings.
\newblock In: MFCS, volume 452. 1990 pp. 26--47.
\newblock \doi{10.1007/BFb0029594}.

\bibitem{kozen-1992}
Kozen DC.
\newblock Design and Analysis of Algorithms.
\newblock Texts and Monographs in Computer Science. Springer, 1992.
\newblock ISBN 978-3-540-97687-5.
\newblock \doi{10.1007/978-1-4612-4400-4}.

\bibitem{kozen-1994}
Kozen D.
\newblock A Completeness Theorem for {K}leene Algebras and the Algebra of Regular Events.
\newblock \emph{Inf. Comput.}, 1994.
\newblock \textbf{110}(2):366--390.
\newblock \doi{10.1006/inco.1994.1037}.

\bibitem{kozen-1996}
Kozen D.
\newblock Kleene algebra with tests and commutativity conditions.
\newblock In: TACAS. 1996 pp. 14--33.
\newblock \doi{10.1007/3-540-61042-1_35}.

\bibitem{kozen-1997}
Kozen D.
\newblock {K}leene Algebra with Tests.
\newblock \emph{{ACM} Trans. Program. Lang. Syst.}, 1997.
\newblock \textbf{19}(3):427--443.
\newblock \doi{10.1145/256167.256195}.

\bibitem{kozen-2001}
Kozen D.
\newblock Myhill-{N}erode Relations on Automatic Systems and the Completeness of {K}leene Algebra.
\newblock In: STACS. 2001 pp. 27--38.
\newblock \doi{10.1007/3-540-44693-1_3}.

\bibitem{kozen-2006}
Kozen D.
\newblock On the Representation of Kleene Algebras with Tests.
\newblock In: MFCS. 2006 pp. 73--83.
\newblock \doi{10.1007/11821069_6}.

\bibitem{kozen-patron-2000}
Kozen D, Patron M.
\newblock Certification of Compiler Optimizations Using {K}leene Algebra with Tests.
\newblock In: CL. 2000 pp. 568--582.
\newblock \doi{10.1007/3-540-44957-4_38}.

\bibitem{kozen-silva-2020}
Kozen D, Silva A.
\newblock Left-handed completeness.
\newblock \emph{Theor. Comput. Sci.}, 2020.
\newblock \textbf{807}:220--233.
\newblock \doi{10.1016/j.tcs.2019.10.040}.

\bibitem{kozen-smith-1996}
Kozen D, Smith F.
\newblock Kleene Algebra with Tests: Completeness and Decidability.
\newblock In: CSL. 1996 pp. 244--259.
\newblock \doi{10.1007/3-540-63172-0_43}.

\bibitem{kozen-tiuryn-2003}
Kozen D, Tiuryn J.
\newblock Substructural logic and partial correctness.
\newblock \emph{{ACM} Trans. Comput. Log.}, 2003.
\newblock \textbf{4}(3):355--378.
\newblock \doi{10.1145/772062.772066}.

\bibitem{kozen-tseng-2008}
Kozen D, Tseng WD.
\newblock The {B}{\"{o}}hm-{J}acopini Theorem Is False, Propositionally.
\newblock In: MPC. 2008 pp. 177--192.
\newblock \doi{10.1007/978-3-540-70594-9_11}.

\bibitem{krob-1990}
Krob D.
\newblock A Complete System of {B}-Rational Identities.
\newblock In: ICALP. 1990 pp. 60--73.
\newblock \doi{10.1007/BFb0032022}.

\bibitem{laurence-struth-2014}
Laurence MR, Struth G.
\newblock Completeness Theorems for Bi-Kleene Algebras and Series-Parallel Rational Pomset Languages.
\newblock In: RAMiCS. 2014 pp. 65--82.
\newblock \doi{10.1007/978-3-319-06251-8_5}.

\bibitem{laurence-struth-2017-arxiv}
Laurence MR, Struth G.
\newblock Completeness Theorems for Pomset Languages and Concurrent Kleene Algebras, 2017.
\newblock \eprint{1705.05896}.

\bibitem{mcnaughton-yamada-1960}
McNaughton R, Yamada H.
\newblock Regular Expressions and State Graphs for Automata.
\newblock \emph{{IRE} Trans. Electronic Computers}, 1960.
\newblock \textbf{9}(1):39--47.
\newblock \doi{10.1109/TEC.1960.5221603}.

\bibitem{mcnaughton-papert-1968}
McNaughton R, Papert S.
\newblock The syntactic monoid of a regular event.
\newblock \emph{Algebraic Theory of Machines, Languages, and Semigroups}, 1968.
\newblock pp. 297--312.

\bibitem{palka-2005}
Palka E.
\newblock On Finite Model Property of the Equational Theory of {K}leene Algebras.
\newblock \emph{Fundam. Informaticae}, 2005.
\newblock \textbf{68}(3):221--230.
\newblock \urlprefix\url{http://content.iospress.com/articles/fundamenta-informaticae/fi68-3-02}.

\bibitem{pous-2014}
Pous D.
\newblock Symbolic Algorithms for Language Equivalence and {K}leene Algebra with Tests.
\newblock In: Proc. Principles of Programming Languages ({POPL}). 2015 pp. 357--368.
\newblock \doi{10.1145/2676726.2677007}.

\bibitem{pratt-1980}
Pratt VR.
\newblock Dynamic Algebras and the Nature of Induction.
\newblock In: STOC. 1980 pp. 22--28.
\newblock \doi{10.1145/800141.804649}.

\bibitem{redko-1964}
Redko VN.
\newblock On Defining Relations for the Algebra of Regular Events.
\newblock \emph{Ukrainian Mathematical Journal}, 1964.
\newblock \textbf{16}:120--126.
\newblock (in Russian).

\bibitem{rutten-2000}
Rutten JJMM.
\newblock Universal coalgebra: a theory of systems.
\newblock \emph{Theor. Comput. Sci.}, 2000.
\newblock \textbf{249}(1):3--80.
\newblock \doi{10.1016/S0304-3975(00)00056-6}.

\bibitem{salomaa-1966}
Salomaa A.
\newblock Two Complete Axiom Systems for the Algebra of Regular Events.
\newblock \emph{J. {ACM}}, 1966.
\newblock \textbf{13}(1):158--169.
\newblock \doi{10.1145/321312.321326}.

\bibitem{schmid-kappe-kozen-silva-2021}
Schmid T, Kapp{\'{e}} T, Kozen D, Silva A.
\newblock Guarded {K}leene Algebra with Tests: Coequations, Coinduction, and Completeness.
\newblock In: ICALP. 2021 pp. 142:1--142:14.
\newblock \doi{10.4230/LIPIcs.ICALP.2021.142}.

\bibitem{smolka-eliopoulos-foster-guha-2015}
Smolka S, Eliopoulos SA, Foster N, Guha A.
\newblock A fast compiler for {NetKAT}.
\newblock In: ICFP. 2015 pp. 328--341.
\newblock \doi{10.1145/2784731.2784761}.

\bibitem{smolka-foster-hsu-kappe-kozen-silva-2020}
Smolka S, Foster N, Hsu J, Kapp\'{e} T, Kozen D, Silva A.
\newblock Guarded {K}leene Algebra with Tests: Verification of Uninterpreted Programs in Nearly Linear Time.
\newblock In: POPL. 2020 pp. 61:1--61:28.
\newblock \doi{10.1145/3371129}.

\bibitem{stockmeyer-meyer-1973}
Stockmeyer LJ, Meyer AR.
\newblock Word Problems Requiring Exponential Time: Preliminary Report.
\newblock In: STOC. 1973 pp. 1--9.
\newblock \doi{10.1145/800125.804029}.

\bibitem{streicher-1993}
Streicher T.
\newblock Investigations into Intensional Type Theory.
\newblock \emph{Habilitiation Thesis, Ludwig Maximilian Universit{\"a}t}, 1993.
\newblock \urlprefix\url{https://www2.mathematik.tu-darmstadt.de/~streicher/HabilStreicher.pdf}.

\bibitem{thompson-1968}
Thompson K.
\newblock Regular Expression Search Algorithm.
\newblock \emph{Commun. {ACM}}, 1968.
\newblock \textbf{11}(6):419--422.
\newblock \doi{10.1145/363347.363387}.

\bibitem{watson-1993}
Watson BW.
\newblock A taxonomy of finite automata construction algorithms.
\newblock Technical report, Technische Universiteit Eindhoven, 1993.
\newblock \urlprefix\url{https://research.tue.nl/files/2482472/9313452}.

\end{thebibliography}

\end{document}